\documentclass[10pt,journal]{IEEEtran}

\usepackage{amssymb}
\usepackage{amsmath}
\usepackage{cite}
\usepackage{url}
\usepackage{xcolor}
\usepackage{cite,graphicx,amsmath,amssymb}
\usepackage{fancyhdr}
\usepackage{mdwmath}
\usepackage{mdwtab}
\usepackage{caption}
\usepackage{amsthm}
\usepackage{setspace}
\usepackage{hyperref}
\usepackage{algorithm}
\usepackage{algorithmic}
\usepackage{multirow}
\usepackage{makecell}
\usepackage{subcaption}
\hypersetup{colorlinks=false}

\graphicspath{{./src/img/}}

\newtheorem{remark}{Remark}
\newtheorem{theorem}{Theorem}

\newtheorem{lemma}{Lemma}

\newtheorem{corollary}{Corollary}

\newtheorem{proposition}{Proposition}

\allowdisplaybreaks
\title{Bidirectional Integrated Sensing and Communication: Full-Duplex or Half-Duplex?}

\author{

        Zhaolin~Wang,~\IEEEmembership{Graduate Student Member,~IEEE,}
        Xidong~Mu,~\IEEEmembership{Member,~IEEE,} \\
        and Yuanwei~Liu,~\IEEEmembership{Fellow,~IEEE}

\thanks{An earlier version of this paper was presented in part at the IEEE International Conference on Communications, Rome, Italy, June 2023 \cite{wang2022conference}.
        
The authors are with the School of Electronic Engineering
and Computer Science, Queen Mary University of London, London E1 4NS,
U.K. (e-mail: zhaolin.wang@qmul.ac.uk, xidong.mu@qmul.ac.uk, yuanwei.liu@qmul.ac.uk).}
\vspace{-0.8cm}
}

\begin{document}

\maketitle
\begin{abstract}
    A bidirectional integrated sensing and communication (ISAC) system is proposed, in which a pair of transceivers carry out two-way communication and mutual sensing. 
    Both full-duplex and half-duplex operations in narrowband and wideband systems are conceived for the bidirectional ISAC. 1) For the narrowband system, the conventional full-duplex and half-duplex operations are redesigned to take into account sensing echo signals. Then, the transmit beamforming design of both transceivers is proposed for addressing the sensing and communication (S\&C) tradeoff. A one-layer iterative algorithm relying on successive convex approximation (SCA) is proposed to obtain Karush-Kuhn-Tucker (KKT) optimal solutions. 2) For the wideband system, the new full-duplex and half-duplex operations are proposed for the bidirectional ISAC. In particular, the frequency-selective fading channel is tackled by delay pre-compensation and path-based beamforming. By redesigning the proposed SCA-based algorithm, the KKT optimal solutions for path-based beamforming for characterizing the S\&C tradeoff are obtained. Finally, the numerical results show that: i) For both bandwidth scenarios, \emph{full-duplex mode may not always be preferable to half-duplex mode} due to the presence of the sensing interference; and ii) For both duplex operations, it is sufficient to reuse communication signals for sensing in the narrowband system, while an additional dedicated sensing signal is required in the wideband system.
\end{abstract}

\section{Introduction}
The peak data rate of next-generation wireless networks is expected to be at least one terabit per second \cite{zhang20196g}, which requires very high spectral efficiency. To this end, full-duplex communication, which allows overlap of downlink (DL) and uplink (UL) communications over the same time-frequency resource, is a promising candidate technique in next-generation wireless networks \cite{samsung20206g}. Theoretically, the full-duplex technique is capable of doubling the communication capacity. However, the main obstacle to the full-duplex technique is self-interference (SI) caused by strong power leakage between adjacently deployed transmitter and receiver, which typically overwhelms the desired communication signals \cite{sabharwal2014band}. Thus, SI cancellation techniques are critical for full-duplex techniques. With the advanced SI cancellation techniques in the propagation domain, analog domain, and digital domain \cite{riihonen2011mitigation, bharadia2013full, everett2014passive, zhang2015full, sim2017nonlinear, khaledian2018inherent}, the strong SI can be suppressed to the receiver noise level, which makes the practical implementation of full-duplex communication systems possible. 

Furthermore, the next-generation wireless networks are envisioned not only to enhance the wireless communication performance, but also to provide sensing services to emerging sensing-based applications such as virtual reality, Internet of Vehicles, and Metaverse. Therefore, integrated sensing and communication (ISAC) techniques have been regarded as one of the key enablers of next-generation wireless networks \cite{zhang2020perceptive, liu2022integrated, tong2022environment}. In ISAC, radio-frequency sensing and wireless communication functionalities are carried out simultaneously by sharing the same frequency spectrum and hardware facilities, thus enhancing resource efficiency. Moreover, based on the widely deployed wireless infrastructures, ubiquitous sensing can be realized by exploiting the advanced ISAC techniques \cite{zhang2020perceptive, liu2022integrated, tong2022environment}, which are capable of capturing the environment data for building ubiquitous intelligence and bridging the virtual and physical worlds.

\subsection{Prior Works}

\subsubsection{Studies on Full-duplex Communication}
The existing research contributions of full-duplex communication may be classified into two categories: full-duplex for DL/UL communication \cite{day2012full, taghizadeh2018hardware, da2016spectral, sun2016multi} and full-duplex for relaying \cite{day2012full_relay, ng2012dynamic, suraweera2014low, ngo2014multipair}. Firstly, for the full-duplex DL/UL communication, the authors of \cite{day2012full} conceived a classical bidirectional topology, where a pair of full-duplex terminals communicate with each other. Moreover, an explicit residual SI model under the limited dynamic range was proposed. As a further step, the authors of \cite{taghizadeh2018hardware} extended the full-duplex bidirectional communication system into the wideband orthogonal frequency division multiplexing (OFDM) systems, where the residual SI model in the frequency-domain was derived. By considering the DL/UL communication in the cellular networks, the authors of \cite{da2016spectral} jointly optimized the frequency channel pairing and power allocation to guarantee spectral efficiency and user fairing. Moreover, multi-objective optimization was invoked in \cite{sun2016multi} for jointly minimizing the downlink and uplink transmit power in the full-duplex communication systems. Next, for the full-duplex relay, a classical full-duplex relay topology with a pair of transmitter and receiver was studied in \cite{day2012full_relay}. Furthermore, the authors of \cite{ng2012dynamic} proposed a full-duplex hybrid relay system, where the resource allocation and scheduling are jointly optimized. For maximizing the end-to-end signal-to-noise ratio (SNR) of the full-duplex relay system, several antenna selection schemes were designed to reduce the system complexity. Finally, the authors of \cite{ngo2014multipair} investigated a multipair full-duplex relay system, where the optimal power allocation for maximizing energy efficiency was obtained.

\subsubsection{Studies on ISAC}
The transmit waveform design plays a key role in implementing dual functions in ISAC systems, which has attracted extensive attention \cite{liu2018toward, liu2020joint, pritzker2022transmit, zhang2022holographic}. The authors of \cite{liu2018toward} reused the communication signals for carrying out target sensing and investigated the sensing and communication (S\&C) tradeoff under several design criteria of sensing beampattern. To achieve the full degrees-of-freedom (DoFs) of sensing, the authors of \cite{liu2020joint} proposed to exploit the composite transmit signal, where an additional dedicated sensing signal is added together with the communication signals. It was shown that the composite transmit signal is capable of achieving better sensing beampattern than the communication-only signal. As a further advance, a flexible beamforming approach is proposed in \cite{pritzker2022transmit} for guaranteeing the desired levels of S\&C performance. Furthermore, the authors of \cite{zhang2022holographic} developed a holographic beamforming scheme that employs more densely deployed radiation elements in an antenna array to realize finer controllability of the S\&C beams. However, only considering the waveform design at the transmitter cannot catch the overall sensing performance. As such, some works also considered the sensing performance metrics at the receiver \cite{liu2022joint, liu2021cramer, wang2022stars}. Specifically, the authors of \cite{liu2022joint} jointly optimized the transmit waveform and the filter at the receiver to maximize the signal-to-interference-plus-noise ratio (SINR) of the sensing echo signal. To characterize the parameter estimation accuracy at the receiver, the fundamental Cram{\'e}r-Rao bound (CRB) was exploited as the sensing performance metric in \cite{liu2021cramer} and \cite{wang2022stars}. Most recently, there are growing research contributions to the new modulation techniques for ISAC. For example, the authors of \cite{li2022novel} conceived an ISAC framework based on the orthogonal time frequency space (OTFS) modulation technique. In OTFS, the communication symbols are multiplexed in the delay-Doppler domain, which well matches the parameter estimation in the sensing function. Moreover, the authors of \cite{xiao2022integrated} proposed a novel delay alignment modulation (DAM)-aided ISAC framework to guarantee high Doppler shift tolerance and low peak-to-average-power ratio (PAPR).

\subsection{Motivations and Contributions}

As mentioned above, full-duplex communication has been studied in diverse scenarios, but there is still a paucity of research contributions on integrating the sensing function into full-duplex communication systems. It is well known that compared with half-duplex mode, full-duplex mode can almost double the communication capacity by utilizing advanced SI cancellation techniques. However, we note that in full-duplex communication systems, the SI consists of two components, namely direct-path interference and reflected-path interference \cite{sabharwal2014band}. When integrating the sensing function into the full-duplex communication systems, the reflected-path interference cannot be totally eliminated since it also includes the useful sensing echo signal reflected by the target of interest. In other words, part of the SI, i.e., the sensing echo signal, needs to be preserved to guarantee the sensing performance. In this case, an interesting question arises, \emph{does full-duplex mode with the preserved sensing echo signal still outperform half-duplex mode?} 

To answer this question, we propose to integrate the sensing function into the classical bidirectional communication system \cite{day2012full, sabharwal2014band}, which is referred to as a bidirectional ISAC system. In this system, a pair of dual-functional transceivers carry out two-way communication and mutual sensing. Then, we investigate the full-duplex and half-duplex operations for the bidirectional ISAC system in both narrowband and wideband scenarios. The corresponding transmit beamforming is optimized for characterizing S\&C tradeoff regions achieved by both full-duplex and half-duplex operations, which are compared to provide answers to the raised question. It is suggested that \emph{full-duplex mode may not always outperform half-duplex mode}.

The primary contributions of this paper are as follows:

\begin{itemize}
    \item We propose a bidirectional ISAC system, where a pair of transceivers communicate with each other and sense each other's direction. Then, we design the corresponding full-duplex and half-duplex operation protocols, which are distinguished by whether the communication signals are transmitted and received at the same time or not, for both narrowband and wideband systems. Based on these protocols, we further study the S\&C tradeoff by optimizing the transmit beamforming.
    \item For the narrowband system, we redesign the conventional full-duplex and half-duplex operations for the bidirectional ISAC system. Then, for both operations, we formulate a joint beamforming optimization problem for maximizing the weighted sum of the communication achievable rate and the sensing CRB, which characterizes the S\&C tradeoff. To solve it, we develop a one-layer successive convex approximation (SCA)-based algorithm to obtain Karush-Kuhn-Tucker (KKT) optimal solutions.
    \item For the wideband system, we propose the new full-duplex and half-duplex operation protocols by exploiting different delays in communication and sensing echo signals. Moreover, the DAM technique is used to address frequency-selective communication channels based on delay pre-compensation and path-based beamforming. Furthermore, we redesign the proposed SCA-based algorithm to obtain the KKT optimal solutions to the S\&C tradeoff optimization problem.
    \item Our numerical results reveal that when the line-of-sight (LOS) component of the communication channel is not dominated, full-duplex and half-duplex operations are superior in the communication-prior regime and sensing-prior regime, respectively. When the communication channels become LOS-dominated, half-duplex always outperforms full-duplex due to the strong interference caused by sensing. It is also shown that the proposed wideband system requires an additional dedicated sensing signal for fulfilling the sensing performance, but reusing the communication signal is sufficient in the narrowband system.
\end{itemize}

\subsection{Organization and Notations}

The rest of this paper is organized as follows. Section \ref{sec:narrowband} presents the half-duplex and full-duplex transmission schemes for the narrowband bidirectional ISAC system and the S\&C tradeoff optimization problem. Then, an SCA-based algorithm is proposed to obtain the KKT optimal solutions to this problem. In Section \ref{sec:wideband}, the half-duplex and full-duplex transmission schemes for the wideband bidirectional ISAC system are proposed. The KKT optimal solution to the corresponding S\&C tradeoff optimization problem is obtained by redesigning the SCA-based algorithm. Our numerical results are presented in Section \ref{sec:results} comparing the half-duplex and the full-duplex, which is followed by our conclusions
in Section \ref{sec:conslution}.

\emph{Notations:}
Scalars, vectors, and matrices are denoted by the lower-case, bold-face lower-case, and bold-face upper-case letters, respectively; 
$\mathbb{C}^{N \times M}$ and $\mathbb{R}^{N \times M}$ denotes the space of $N \times M$ complex and real matrices, respectively;
$a^*$ and $|a|$ denote the conjugate and magnitude of scalar $a$;  
$\mathbf{a}^H$ denotes the conjugate transpose of vector $\mathbf{a}$; 
$\mathrm{diag}(\mathbf{a})$ denotes a diagonal matrix with same value as the vector $\mathbf{a}$ on the diagonal;
$\mathbf{A} \succeq 0$ means that matrix $\mathbf{A}$ is positive semidefinite; 
$\mathrm{tr}(\mathbf{A})$ denote the trace of matrix $\mathbf{A}$;
$\mathbb{E}[\cdot]$ denotes the statistical expectation; 
$\mathrm{Re}\{\cdot\}$ denotes the real component of a complex number;
$\mathcal{CN}(\mu, \sigma^2)$ denotes the distribution of a circularly symmetric complex Gaussian random variable with mean $\mu$ and variance $\sigma^2$.

\section{Narrowband Bidirectional ISAC Systems} \label{sec:narrowband}


\begin{figure}[t!]
      \centering
      \includegraphics[width=0.4\textwidth]{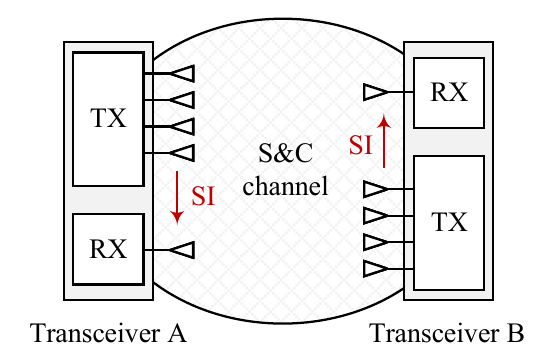}
      \caption{Illustration of the bidirectional ISAC system.}
      \label{fig:system_model}
\end{figure}

We focus our attention on integrating the sensing function into the typical bidirectional communication system \cite{day2012full, sabharwal2014band}, which is referred to as a bidirectional ISAC system. In this system, two dual-functional transceivers, denoted by the set $\mathcal{K} = {A, B}$, have the dual objectives of establishing communication between each other and concurrently sensing each other's direction. Each of these transceivers is equipped with a total of $M$ transmit antennas and a single receive antenna. The bandwidth of this system is denoted by $W$, which corresponds to a symbol period $T_s = 1/W$. In this paper, we consider both narrowband and wideband transmission schemes for the bidirectional ISAC system. The key differences between these two transmission schemes are summarized as follows. On the one hand, in narrowband systems, the communication channel is characterized as frequency-flat, whereas, in wideband systems, the channel exhibits frequency-selective behavior. On the other hand, the delays in both communication and sensing signals can be neglected in the narrowband systems due to the relatively large value of $T_s$. However, this simplification is not applicable in the context of wideband systems. In subsequent sections, we delve into the study of the narrowband bidirectional ISAC system in Section \ref{sec:narrowband}, and we address the more intricate wideband system in Section \ref{sec:wideband}.


\begin{figure}[t!]
  \centering
  \includegraphics[width=0.5\textwidth]{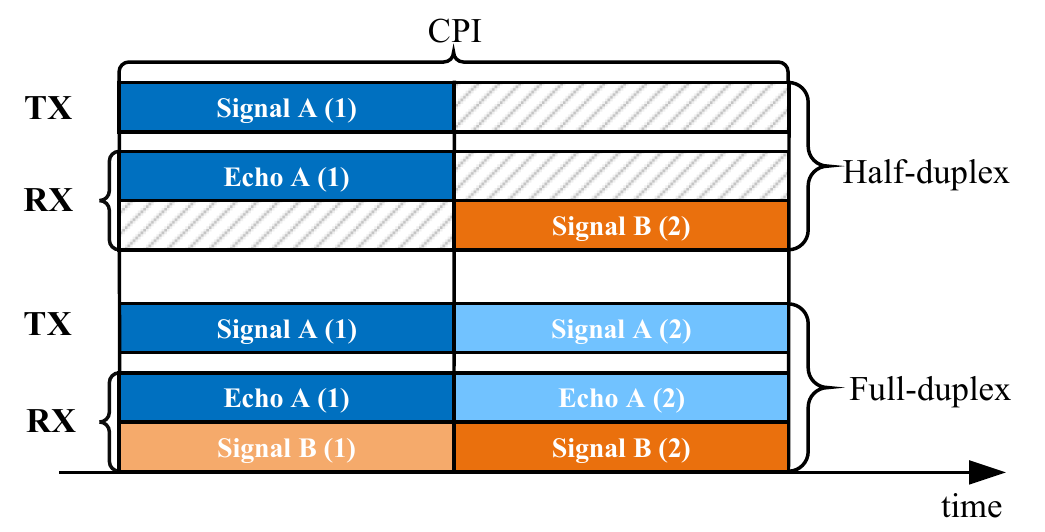}
  \caption{Half-/full-duplex protocols for narrowband system (the view from transceiver $A$).}
  \label{fig:protocols_narrow}
\end{figure}

In the narrowband system,  the integration of a sensing function into the established bidirectional communication system [cite{day2012full, sabharwal2014band}] allows for the direct redesign of corresponding full-duplex and half-duplex operation protocols, as depicted in Fig. \ref{fig:protocols_narrow}. Specifically, we focus on one coherent processing interval (CPI) of length $2 N$,  during which communication channels and target parameters remain approximately constant. For the half-duplex mode, transceiver $A$ transmits the joint S\&C signal to transceiver $B$ in the first half of the CPI. Subsequently, both Transceivers B and A, during the same time period, receive the communication signal and the sensing echo signal, respectively. Following this, transceiver $B$ assumes the role of transmitting the joint S\&C signal to transceiver $A$ during the second half of the CPI, repeating a similar sequence of events. For the full-duplex mode, transceivers $A$ and $B$ transmit the joint S\&C signals simultaneously throughout the entire CPI. Thus, each transceiver receives superimposed communication and sensing echo signals. 

Based on the above operating protocol, we divide the CPI into two time intervals of equal duration, denoted as $\mathcal{P}_1 = \{1,...,N\}$ and $\mathcal{P}_2 = \{N+1,...,2N\}$. Within this context, we denote $\mathbf{x}_{k,i}[n] \in \mathbb{C}^{M \times 1}$ as the normalized joint S\&C signal transmitted by the transceiver $k$ at time $n$ within the interval $\mathcal{P}_i$. The transmit signal covariance matrix is thereby defined as $\mathbf{Q}_{k,i} = \mathbb{E}[\mathbf{x}_{k,i}[n] \mathbf{x}_{k,i}^H[n] ] \in \mathbb{C}^{M \times M}$, with the constraint $\mathbf{Q}_{k,i} \succeq 0$ ensuring that it is a positive semidefinite matrix. We further consider an average power constraint that applies across the entire CPI at each transceiver, and it is expressed as follows:
\begin{equation}
  \sum_{i =1}^2 \mathbb{E} \left[\| \mathbf{x}_{k,i}[n] \|^2 \right] = \sum_{i =1}^2 \mathrm{tr}(\mathbf{Q}_{k,i}) \le 2.
\end{equation}
To facilitate the full-duplex and half-duplex modes, the following constraints must be met in the narrowband system: 
\begin{subequations}
  \begin{align}
    \label{eqn:narrow_mode_1}
    &\text{Full-duplex:} \quad \mathbf{Q}_{A,1} = \mathbf{Q}_{A,2}, \mathbf{Q}_{B,1} = \mathbf{Q}_{B,2}, \\
    \label{eqn:narrow_mode_2}
    &\text{Half-duplex:} \quad \mathbf{Q}_{A,2} = \mathbf{0}_{M \times M}, \mathbf{Q}_{B,1} = \mathbf{0}_{M \times M}.
  \end{align}
\end{subequations}
In the narrowband system, each communication channel experiences frequency-flat fading, allowing for representation by a single filter tap \cite{tse2005fundamentals}. Consequently, we denote $\mathbf{h}_k \in \mathbb{C}^{M \times 1}$ as the complex baseband communication channel vector originating from transceiver $k$ to the other transceiver. Furthermore, $\mathbf{g}_k \in \mathbb{C}^{M \times 1}$ characterizes the complex baseband SI channel vector at transceiver $k$, encompassing both direct-path interference and reflected-path interference generated by objects in the environment other than the target.
The received signal of transceiver $k$ at time $n$ within $\mathcal{P}_i$ can be expressed as
\begin{align} \label{eqn:downlink_signal}
  y_{k,i}[n] = \underbrace{ \sqrt{\rho_c} \mathbf{h}_{k'}^H \mathbf{x}_{k',i}[n]}_{\text{communication signal}} & + \underbrace{\sqrt{\rho_s} \alpha_k e^{j 2 \pi \nu_k n T_s } \mathbf{a}^H(\theta_k) \mathbf{x}_{k,i}[n]}_{\text{sensing echo signal}} \nonumber \\ & + \underbrace{ \sqrt{\eta}\mathbf{g}_k^H \mathbf{x}_{k,i}[n]}_{\text{SI}} + n_{k,i}[n].
\end{align}
Here, $k' = B$ if $k = A$ and $k' = A$ otherwise. Moreover, $\rho_c > 0$ and $\rho_s > 0$ represent the signal-to-noise ratios (SNR) of the communication and sensing echo signals, respectively, while $\eta > 0$ denotes the interference-to-noise ratio (INR) of the SI, determined by antenna separation and analog-domain suppression. Additionally, $n_k[n] \sim \mathcal{CN}(0, 1)$ characterizes the normalized additive white Gaussian noise at the receiver. Within the sensing echo signal, $\alpha_k \in \mathbb{C}$ signifies the complex reflection coefficient following the Swerling-II model \cite{skolnik1980introduction}. Moreover, $\nu_k = 2 v_k / \lambda$ represents the Doppler frequency shift introduced by the relative radial velocity $v_k$, with $\lambda$ denoting the signal wavelength. Furthermore, $\mathbf{a}(\theta_k) \in \mathbb{C}^{M \times 1}$ is the steering vector of the transmitter at angle $\theta_k$. In this paper, we assume that the transmitters are equipped with uniform linear arrays (ULAs). The corresponding steering vector under the far-field planar wave assumption can be defined as:
\begin{align}
  &\mathbf{a}(\theta_k) = \left[1, e^{j\frac{2\pi}{\lambda} d \sin(\theta_k)},...,e^{j\frac{2\pi}{\lambda} (M-1)d \sin(\theta_k) } \right]^T,
\end{align}
where $d$ denotes the antenna spacing.

To ensure the robust performance of the ISAC system, it is imperative to address the adverse effects of SI. This necessitates the availability of precise channel state information (CSI) specific to the SI channel. Although the sensing channel and the SI channel are inherently intertwined, it is still possible to distinguish between them by exploiting their distinct Doppler frequency shifts. However, even with the perfect CSI of the SI channel, the complete elimination of SI may remain unattainable due to the limited dynamic range of the receiver. As a result, akin to the approach presented in \cite{day2012full}, the residual SI, stemming from the finite receiver dynamic range, is modeled as independent Gaussian noise characterized by a mean of zero and a variance proportionally linked to the received signal power. For the sake of clarity, we will focus primarily on the influence of the SI channel on the residual SI, given its much higher power compared to other components. In this case, according to \cite[Eq. (4)]{day2012full}, the residual SI at transceiver $k$ can be characterized by $e_{k,i}[n] \sim \mathcal{CN}(0, \rho_{\text{SI}} \Phi_{k,i})$, where $0 < \rho_{\text{SI}} \ll 1$ denotes the scaling factor that depends on the receiver's dynamic range. Additionally, $\Phi_{k,i} = \eta \mathbf{g}_k^H \mathbf{Q}_{k,i} \mathbf{g}_k$. Subsequently, the received signal by transceiver $k$ at time $n$ within $\mathcal{P}_i$, considering the presence of residual SI, can be expressed as follows:
\begin{align}
  \tilde{y}_{k,i}[n] &= \sqrt{\rho_c} \mathbf{h}_{k'}^H \mathbf{x}_{k',i}[n] \nonumber \\ &+ \sqrt{\rho_s} \alpha_k e^{j 2 \pi \nu_k n T_s } \mathbf{a}^H(\theta_k) \mathbf{x}_{k,i}[n] + e_{k,i}[n] + n_{k,i}[n].
\end{align}

\subsection{Transmission and Reception Schemes}
In this subsection, we present the transmission and reception schemes for the narrowband bidirectional ISAC system. The perfect CSI is assumed at the transceivers. 
Although it has been shown in \cite{liu2018toward} that the communication signal can be reused for target sensing, recent works have proved that the dedicated sensing signal is required in the transmitted joint S\&C signals to achieve the full degrees-of-freedom of the target sensing \cite{liu2020joint, pritzker2022transmit, zhang2022holographic, liu2022integrated}. Therefore, the transmit signal $\mathbf{x}_{k,i}[n]$ can be designed as follows:
\begin{equation}
  \mathbf{x}_{k,i}[n] = \mathbf{w}_{k,i} c_{k,i}[n] + \mathbf{s}_{k,i}[n],
\end{equation}
where $c_{k,i}[n] \sim \mathcal{CN}(0,1)$ denotes the independent unit-power information symbol transmitted by transceiver $k$ to the other, $\mathbf{w}_{k,i} \in \mathbb{C}^{M \times 1}$ denotes the corresponding transmit beamformer, and $\mathbf{s}_{k,i}[n] \in \mathbb{C}^{M \times 1}$ denotes the dedicated sensing signal generated by the pseudo-random coding. Denote the covariance of the dedicated sensing signal by $\mathbf{R}_{k,i} = \mathbb{E}[\mathbf{s}_{k,i}[n] \mathbf{s}^H_{k,i}[n]] \in \mathbb{C}^{M \times M}$, with the constraint $\mathbf{R}_{k,i} \succeq 0$. Then, the covariance matrix of the transmit signal $\mathbf{x}_{k,i}[n]$ can be explicitly written as
\begin{equation}
  \mathbf{Q}_{k,i} = \mathbf{w}_{k,i} \mathbf{w}^H_{k,i} + \mathbf{R}_{k,i} \succeq \mathbf{w}_{k,i} \mathbf{w}_{k,i}^H.
\end{equation} 
Note that since the signal $\mathbf{s}_{k,i}[n]$ is dedicated to sensing and does not contain random communication information, it can be \emph{a-priori} designed and known at the other transceiver prior to transmission \cite{liu2022integrated}. As such, the interference caused by the dedicated sensing signal can be pre-canceled at the receiver. The desired information symbol $c_{k',i} [n]$ transmitted from transceiver $k'$ to transceiver $k$ can be decoded through the following signal:
\begin{align} \label{eqn:comm_signal_narrow}
  \bar{y}_{k,i}[n] = &\tilde{y}_{k,i}[n] - \sqrt{\rho_c} \mathbf{h}_{k'}^H \mathbf{s}_{k',i}[n] \nonumber \\
  = &\sqrt{\rho_c} \mathbf{h}_{k'}^H \mathbf{w}_{k',i} c_{k',i} [n] + \bar{n}_{k,i}[n],
\end{align}
where $\bar{n}_{k,i}[n]$ denotes the remaining interference plus noise as follows: 
\begin{equation}
  \bar{n}_{k,i}[n] = \sqrt{\rho_s} \alpha_k e^{j 2 \pi \nu_k n T_s } \mathbf{a}^H(\theta_k) \mathbf{x}_{k,i}[n] + e_{k,i}[n] + n_{k,i}[n].
\end{equation}  
After decoding the information symbol $c_{k',i} [n]$, we propose to adopt the successive interference cancellation (SIC) technique to remove the interference from the communication signal to the sensing echo By assuming the perfect SIC, the effective sensing signal is given by
\begin{align}
  \psi_{k,i}[n] &= \bar{y}_{k,i}[n] - \sqrt{\rho_c} \mathbf{h}_{k'}^H \mathbf{w}_{k',i} c_{k',i} [n] \nonumber \\ & = \sqrt{\rho_s} \alpha_k e^{j 2 \pi \nu_k n T_s } \mathbf{a}^H(\theta_k) \mathbf{x}_{k,i}[n] + z_{k,i}[n],
\end{align}
where $z_{k,i}[n] = e_{k,i}[n] + n_{k,i}[n]$ denotes the remaining residual SI interference plus noise. Then, the angle $\theta_k$ can be estimated from all the samples $\psi_{k,i}[n]$ across the entire CPI through classical estimation schemes such as maximum likelihood estimation. The estimated $\theta_k$ can be used for some location-based applications and to promote the communication channel estimation in the next CPI.

\subsection{Performance Metrics}

\subsubsection{Achievable Rate of Communication} 
To evaluate the communication performance, we exploit the classical achievable rate as the performance metric, which is determined by the SINR of the received signal at each transceiver. According to the effective communication signal at transceiver $k$ in \eqref{eqn:comm_signal_narrow}, the SINR is given by
\begin{equation} \label{eqn:communication_SINR}
  \gamma_{k,i} = \frac{ \rho_c \left|\mathbf{h}_{k'}^H \mathbf{w}_{k',i} \right|^2}{ \rho_s \mathbf{a}^H(\theta_k) \mathbf{Q}_{k,i} \mathbf{a}(\theta_k) + \rho_{\text{SI}}\Phi_{k,i} + 1 }.
\end{equation}
Therefore, the average achievable communication rate at transceiver $k$ is given by 
\begin{equation} \label{eqn:rate}
  R_k = \frac{1}{2} \sum_{i =1}^2 \log_2 \left( 1 + \gamma_{k,i}\right).
\end{equation} 

\subsubsection{Cram{\'e}r-Rao Bound of Sensing}
While characterizing the Mean Square Error (MSE) of the estimated directions remains a complex task, the closed-form Cramer-Rao Bound (CRB) emerges as a valuable tool, offering a lower bound on the MSE achievable by unbiased estimators \cite{kay1993fundamentals, liu2021cramer}. Consequently, we adopt the CRB as the metric of choice for evaluating the sensing performance. To derive the CRB for estimating the direction $\theta_k$, we begin by consolidating all the samples of the effective sensing signal across the entire CPI into a single vector, denoted as $\boldsymbol{\psi}_k$, which can be represented as follows:
\begin{align}
  \boldsymbol{\psi}_k = &\big[\psi_{k,1}[1],...,\psi_{k,1}[N], \psi_{k,2}[N+1],...,\psi_{k,2}[2N] \big]^T \nonumber \\ = &\sqrt{\rho_s} \alpha_k (\mathbf{I}_{2N} \otimes \mathbf{a}^H(\theta_k)) \mathrm{vec}\left(\mathbf{X}_k \mathrm{diag}( \mathbf{d}(\nu_k) ) \right) + \mathbf{z}_k,
\end{align}
Here, $\mathbf{X}_k \in \mathbb{C}^{M \times 2N}$ represents the composite transmit signal matrix for one CPI, $\mathbf{d}(\nu_k) = [e^{j 2 \pi \nu_k T_s},\dots,e^{j 2 \pi \nu_k 2N T_s}]^T \in \mathbb{C}^{2N \times 1}$ denotes the Doppler frequency shift vector, and $\mathbf{z}_k \in \mathbb{C}^{2N \times 1}$ stands for the interference plus noise vector. Following \cite{kay1993fundamentals}, the expected CRB for estimating $\theta_k$ from $\boldsymbol{\psi}_k$ is derived in Appendix A, yielding the following expression:
\begin{align} \label{eqn:CRB}
  C_k = \frac{1}{2 \rho_s N} \left( \sum_{i =1}^2 \frac{|\alpha_k|^2 \dot{\mathbf{a}}^H(\theta_k) \mathbf{Q}_{k,i} \dot{\mathbf{a}}(\theta_k)}{\rho_{\text{SI}} \Phi_{k,i} + 1} \right)^{-1},
\end{align}
where $\dot{\mathbf{a}}(\theta_k)$ denotes the partial derivation of $\mathbf{a}(\theta_k)$ with respect to $\theta_k$.

\begin{remark} \label{remark_1}
  \emph{
  In light of the expressions for achievable rates in \eqref{eqn:communication_SINR} and \eqref{eqn:rate}, by omitting the interference, it can be demonstrated that the full-duplex mode outperforms the half-duplex mode in terms of communication performance. This assertion is supported by the inequality $2 \log_2(1 + \rho_c^\star) \ge \log_2(1 + 2\rho_c^\star)$, where $\rho_c^\star$ represents the maximum value of $\rho_c |\mathbf{h}_{k'}^H \mathbf{w}_{k',i}|^2$ in the full-duplex mode.
  However, when we consider the CRB, as articulated in \eqref{eqn:CRB}, and exclude interference, the optimal CRB achieved by both duplex mode can be consistently expressed as $(2N\rho_s^\star)^{-1}$, where $\rho_s^\star$ signifies the maximum value of $\rho_s \dot{\mathbf{a}}^H(\theta_k) \mathbf{Q}_{k,i} \dot{\mathbf{a}}(\theta_k)$ in the full-duplex mode. Consequently, it becomes apparent that the full-duplex mode does not yield a significant enhancement in sensing performance, a notion that will be further corroborated by our numerical results.
  } 
\end{remark}

\begin{remark} \label{remark_2}
  \emph{
    In \textbf{Remark \ref{remark_1}}, we conclude that the full-duplex mode can attain superior communication performance compared to the half-duplex mode under the condition of negligible interference. However, it is essential to emphasize that even if the SI is effectively suppressed, the interference introduced by the sensing echo signal, which is critical for ensuring sensing performance, must be retained. Consequently, the sensing echo signal assumes the role of another category of “unsuppressed SI” with respect to communication. This implies that the full-duplex mode may not necessarily outperform the half-duplex mode, particularly when sensing takes precedence over communication.
  }
\end{remark}

\subsection{Joint Beamforming Design for S\&C Tradeoff} \label{sec:beamforming_narrow}

\subsubsection{Problem Formulation}
To achieve optimal performance in both S\&C, our objective is to maximize the average achievable rate denoted as $R = \sum_{k=1}^2 R_k$ across the two transceivers, while simultaneously minimizing the average Cramer-Rao Bound (CRB) denoted as $C = \sum_{k=1}^2 C_k$. Maximizing the rate $R$ necessitates directing a significant portion of the power towards the desired communication direction and suppressing the sensing echo signal. However, this strategy may not be conducive for minimizing the CRB, as improved estimation accuracy typically demands a stronger sensing echo signal. Consequently, the achievable rate $R$ and CRB $C$ represent conflicting performance metrics, giving rise to the S\&C tradeoff.
In this subsection, our objective is to jointly design the communication beamformers and dedicated sensing signals in order to characterize the S\&C tradeoff. The corresponding optimization problem is a multiple-objective optimization problem (MOOP) with two distinct objectives, namely, the rate $R$ and the CRB $C$. To handle this, the scalarization method is employed to transform the MOOP into a single-objective optimization (SOOP) \cite{marler2004survey}. One of the most commonly used scalarization methods is the weighted sum method, which results in the following SOOP:
\begin{subequations} \label{problem:narrow}
  \begin{align}
      \hspace{-1cm} \max_{\mathbf{w}_{k,i}, \mathbf{Q}_{k,i}} \quad & w R - (1-w) \mu C  \\
      \label{constraint:P1_1}
      \mathrm{s.t.} \quad & \eqref{eqn:narrow_mode_1} \text{ or } \eqref{eqn:narrow_mode_2}, \\
      \label{constraint:P1_2}
      & \sum_{i =1}^2 \mathrm{tr}(\mathbf{Q}_{k,i}) \le 2, \forall k, \\
      \label{constraint:P1_3}
      & \mathbf{Q}_{k,i} \succeq \mathbf{w}_{k,i} \mathbf{w}_{k,i}^H, \forall k, i.
  \end{align}
\end{subequations}
Here, the weight parameter $w \in [0,1]$ characterizes the trade-off between communication and sensing performance priorities, while $\mu > 0$ is a scaling parameter that ensures the comparability of the value of $C$ to $R$. Constraint \eqref{constraint:P1_1} determines the duplex mode and constraint \eqref{constraint:P1_2} enforces an average power constraint. It is important to note that the objective function of problem \eqref{problem:narrow} has no physical meaning since the achievable rate and CRB have different units. Nevertheless, maximizing this objective function remains meaningful, as it can facilitate the joint design of the communication beamformers and dedicated sensing signals under different priorities of S\&C performance by adjusting the weight parameter $w$. Furthermore, in problem \eqref{problem:narrow}, the optimization focuses on the communication beamformer $\mathbf{w}_{k,i}$ and the transmit signal covariance matrix $\mathbf{Q}_{k,i},$ indirectly affecting the dedicated sensing signal $\mathbf{s}_{k,i}[n]$ through constraint \eqref{constraint:P1_3}. Upon obtaining the optimal $\mathbf{w}_{k,i}$ and $\mathbf{Q}_{k,i}$, the optimal covariance matrix of the dedicated sensing signal $\mathbf{s}_{k,i}[n]$ can be calculated using $\mathbf{R}_{k,i} = \mathbf{Q}_{k,i} - \mathbf{w}_{k,i} \mathbf{w}_{k,i}^H$. Further, employing matrix decomposition methods, such as eigenvalue decomposition or Cholesky decomposition \cite{liu2020joint}, allows the construction of the signal $\mathbf{s}_{k,i}[n]$.

\subsubsection{SCA-based Algorithm for Solving Problem \eqref{problem:narrow}}
Problem \eqref{problem:narrow} challenging to solve due to the presence of non-convex fractional terms in the achievable rate $R$ and the CRB $C$.  To address this complexity and simplify the optimization, we introduce a set of auxiliary variables $r_{k,i} \ge 0, \forall k, i$, such that they satisfy the constraint:
\begin{equation}
  \gamma_{k,i} \ge r_{k,i},
\end{equation}
Additionally, we introduce auxiliary variables $d_{k,i} > 0$ and $g_{k,i} > 0, \forall k, i$, which satisfy:
\begin{align}
  G_{k,i} &\triangleq |\alpha_k|^2 \dot{\mathbf{a}}^H(\theta_k) \mathbf{Q}_{k,i} \dot{\mathbf{a}}(\theta_k) \ge g_{k,i}^2, \\ D_{k,i} &\triangleq \frac{g_{k,i}^2}{\rho_{\text{SI}} \Phi_{k,i} + 1} \ge d_{k,i}.
\end{align}
With these auxiliary variables, problem \eqref{problem:narrow} can be reformulated into the following equivalent optimization problem:
\begin{subequations} \label{problem:narrow_1}
  \begin{align}
     \hspace{-1cm} \max_{\scriptstyle \mathbf{w}_{k,i}, \mathbf{Q}_{k,i} \atop \scriptstyle r_{k,i}, g_{k,i}, d_{k,i} }  & f(w, r_{k,i}, d_{k,i}) \\
      \label{constraint:non_convex_narrow_1}
      \mathrm{s.t.} \quad & \gamma_{k,i} \ge r_{k,i}, \forall k, i,\\
      \label{constraint:non_convex_narrow_2}
      &G_{k,i} \ge g_{k,i}^2, \forall k, i\\
      \label{constraint:non_convex_narrow_3}
      & D_{k,i} \ge d_{k,i}, \forall k, i,\\
      & \eqref{constraint:P1_1}-\eqref{constraint:P1_3},
  \end{align}
\end{subequations}
where 
\begin{align}
  &f(w, r_{k,i}, d_{k,i}) = \nonumber \\ & \frac{w}{2} \sum_{k =1}^2 \sum_{i =1}^2 \log_2(1 + r_{k,i}) - \frac{(1-w) \mu}{2 \rho_s N} \sum_{k =1}^2 \left(\sum_{i =1}^2 d_{k,i} \right)^{-1}.
\end{align}
In this reformulated problem, the objective function is now convex. Furthermore, utilizing the Schur complement condition, constraint \eqref{constraint:P1_3} can be transformed into an equivalent convex form as follows:
\begin{subequations} 
  \begin{gather}
    \begin{bmatrix} \label{constraint:Schur_1}
      \mathbf{Q}_{k,i} & \mathbf{w}_{k,i} \\
      \mathbf{w}_{k,i}^H & 1
    \end{bmatrix} \succeq 0, \forall k, i.
  \end{gather}
\end{subequations}
This leaves us with the non-convex constraints \eqref{constraint:non_convex_narrow_1} and \eqref{constraint:non_convex_narrow_3} to address. We employ a one-layer iterative algorithm based on SCA to solve these constraints. The key idea is to use the Taylor expansion of the joint convex quadratic-over-linear function $a^2/b$ to obtain lower bounds on the non-convex functions in \eqref{constraint:non_convex_narrow_1} and \eqref{constraint:non_convex_narrow_3}. 
In particular, given a fixed point $(a^{[t]}, b^{[t]})$ obtained in the $t$-th iteration of SCA, a lower bound of $a^2/b$ can be obtained by the first-order Taylor expansion as follows:
\begin{equation} \label{eqn:SCA}
  \frac{a^2}{b} \ge \frac{2 a^{[t]} }{b^{[t]}}a - \left( \frac{a^{[t]} }{b^{[t]}} \right)^2 b.
\end{equation}
Based on this lower bound, we can directly compute a concave lower bound for the variable $D_{k,i}$ by substituting $a = g_{k,i}$ and $b = \rho_{\text{SI}} \Phi_{k,i} + 1$, which is given by:
\begin{align}
  D_{k,i}^{[t]} \triangleq \frac{2 g_{k,i}^{[t]} }{\rho_{\text{SI}} \Phi_{k,i}^{[t]} + 1} g_{k,i} - \left( \frac{ g_{k,i}^{[t]} }{\rho_{\text{SI}} \Phi_{k,i}^{[t]} + 1} \right)^2 (\rho_{\text{SI}} \Phi_{k,i} + 1).
\end{align}
Similarly, by defining $I_{k,i} = \rho_s \mathbf{a}^H(\theta_k) \mathbf{Q}_{k,i} \mathbf{a}(\theta_k) + \rho_{\text{SI}}\Phi_{k,i} + 1$ and using $a = \sqrt{\rho_c} \mathbf{h}_{k'}^H \mathbf{w}_{k',i}$ and $b = I_{k,i}$, we can derive a concave lower bound for the variable $\gamma_{k,i}$ as follows:
\begin{align}
  \gamma_{k,i}^{[t]}  \triangleq  \frac{ 2 \rho_c \mathrm{Re} \left\{ \mathbf{w}_{k',i}^{[n],H} \mathbf{h}_{k'} \mathbf{h}_{k'}^H \mathbf{w}_{k',i} \right\} }{I_{k,i}^{[t]}}  - \rho_c \left| \frac{\mathbf{h}_{k'}^H (\mathbf{w}_{k',i})^{[t]}}{I_{k,i}^{[t]}} \right|^2 I_{k,i}.
\end{align}
With these concave lower bounds $D_{k,i}^{[t]}$ and $\gamma_{k,i}^{[t]}$, problem \eqref{problem:narrow_1} can be reformulated as a convex optimization problem:
\begin{subequations} \label{problem:narrow_3}
  \begin{align}
      \hspace{-1cm} \max_{\scriptstyle \mathbf{w}_{k,i}, \mathbf{Q}_{k,i} \atop \scriptstyle r_{k,i}, g_{k,i}, d_{k,i}  } & f(w, r_{k,i}, d_{k,i})  \\
      \mathrm{s.t.} \quad & \gamma_{k,i}^{[t]} \ge r_{k,i}, \forall k, i,\\
      & D_{k,i}^{[t]} \ge d_{k,i}, \forall k, i,\\      
      & \eqref{constraint:P1_1}, \eqref{constraint:P1_2}, \eqref{constraint:non_convex_narrow_2}, \eqref{constraint:Schur_1}.
  \end{align}
\end{subequations}
This problem can be effectively solved using standard interior-point algorithms. The overall algorithm for solving problem \eqref{problem:narrow} is presented in \textbf{Algorithm \ref{alg:SCA_narrow}}.
\begin{algorithm}[tb]
  \caption{SCA-based algorithm for solving problem \eqref{problem:narrow}.}
  \label{alg:SCA_narrow}
  \begin{algorithmic}[1]
      \STATE{Initialize feasible $\big\{\mathbf{w}_{k,i}^{[0]}, \mathbf{Q}_{k,i}^{[0]}, g_{k,i}^{[0]} \big\}$, and set $t=0$.}
      \REPEAT
      \STATE{update $\big\{\mathbf{w}_{k,i}^{[t]}, \mathbf{Q}_{k,i}^{[t]}, g_{k,i}^{[t]} \big\}$ by solving \eqref{problem:narrow_3}}.
      \STATE{$t = t+1$.}
      \UNTIL{the fractional reduction of the objective value falls below a predefined threshold.}
  \end{algorithmic}
\end{algorithm}

\subsubsection{Convergence, Optimality, and Complexity Analysis} \label{sec:alg_narrow}
We now delve into the convergence, optimality, and complexity aspects of \textbf{Algorithm \ref{alg:SCA_narrow}}. Firstly, it is important to note that the solutions obtained in each iteration of the algorithm remain feasible for subsequent iterations. Consequently, the objective value does not decrease during the iterations. As the objective function is bounded due to the average power constraint, the proposed algorithm is guaranteed to converge. Moreover, the optimality of the solution acquired by the algorithm is substantiated by the following proposition:
\begin{proposition} \label{proposition_1}
  \emph{
  Denote $\big\{ \mathbf{w}_{k,i}^{[\infty]}, \mathbf{Q}_{k,i}^{[\infty]} \big\}$ as the converged solution obtained by \textbf{Algorithm \ref{alg:SCA_narrow}} as $t \rightarrow \infty$. Then, $\big\{ \mathbf{w}_{k,i}^{[\infty]}, \mathbf{Q}_{k,i}^{[\infty]} \big\}$ is a KKT point of problem \eqref{problem:narrow}.
  }
\end{proposition}
\begin{proof}
  Please refer to Appendix B.
\end{proof}

\begin{figure*} [t!]
  \centering
  \includegraphics[width=0.9\textwidth]{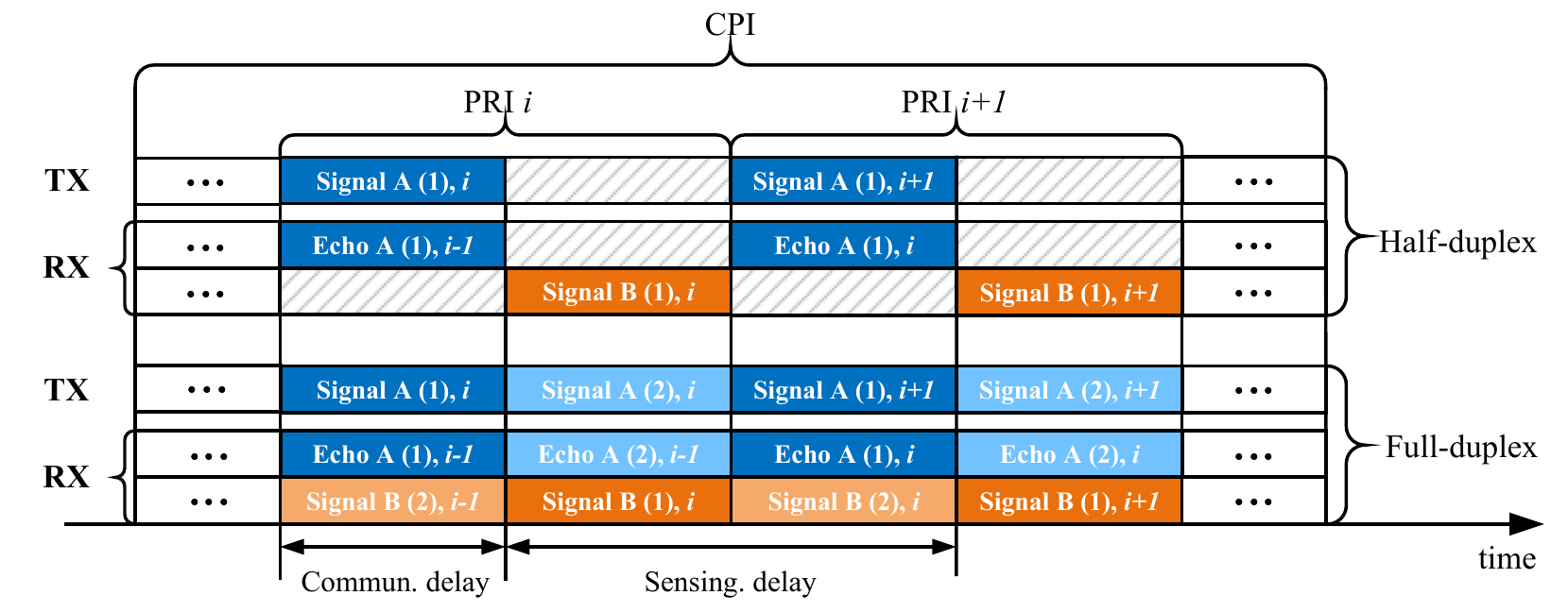}
  \caption{Half-/full-duplex protocols for wideband system (the view from transceiver $A$).}
  \label{fig:protocols_broad}
\end{figure*}

As for the complexity of \textbf{Algorithm \ref{alg:SCA_narrow}}, its worst-case analysis can be presented as follows. The primary source of complexity in the algorithm is associated with solving problem \eqref{problem:narrow_3}. By neglecting the low-dimensional auxiliary variables, problem \eqref{problem:narrow_3} can be viewed as a semidefinite programming (SDP) optimization problem concerning positive semidefinite matrices, as defined in \eqref{constraint:Schur_1}. Given that there are four semidefinite matrices with a dimension $N+1$, the worst-case complexity can be expressed as $\mathcal{O}\left((4^3(N+1)^{3.5} + 4^4)\log(1/\epsilon)\right)$ \cite{huang2010dual}, where $\mathcal{O}(\cdot)$ is the big-O notation and $\epsilon$ denotes the solution accuracy.   
\section{Wideband Bidirectional ISAC Systems} \label{sec:wideband}

In this section, we further study the bidirectional ISAC in the wideband system. Recall that wideband systems present distinct challenges, namely, i) the non-neglected delays in both communication and sensing signals, and ii) the presence of frequency-selective fading channels. To address these challenges, we have introduced novel operational protocols, encompassing both full-duplex and half-duplex modes. These protocols leverage the inherent delays and exhibit a structure akin to their narrowband counterparts. 
Moreover, we incorporate the recently proposed DAM technique \cite{lu2022delay, xiao2022integrated} to address the frequency-selective fading channels. 

The proposed full-duplex and half-duplex operation protocols for the wideband system are illustrated in Fig. \ref{fig:protocols_broad}. In particular, the CPI is divided into a total of $U$ pulse repetition intervals (PRIs). Each PRI has a duration of $N_0 = 2\tau$, where $\tau = \mathrm{round}(\tilde{\tau}/T_s)$ represents the delay of the communication signal, with $\tilde{\tau}$ denoting the delay in second. Owing to the round-trip propagation characteristics, the sensing echo signals exhibit a delay of $2\tau$, which aligns with the duration of one PRI.
For half-duplex mode, the two transceivers jointly transmit S\&C signals during the first half of each PRI, subsequently maintaining a mute state during the latter half to prevent SI. During the first half of each PRI, the transceivers receive the delayed sensing echo signals from the prior PRI. In the second half of the same PRI, they receive the delayed communication signals from the first half of the current PRI.
Conversely, in the full-duplex mode, both transceivers engage in simultaneous transmission of the joint S\&C signals, thereby receiving the superimposed communication and sensing echo signals. It is important to emphasize that, to achieve interference-free transmission in the half-duplex mode, the knowledge of the time delay $\tau$ is required, which can be acquired through matched filters or synchronization information shared between the transceivers.

Similarly, to facilitate the design of full-duplex and half-duplex operation protocols in the wideband system, we partition the $q$-th PRI into two equally-lengthened time intervals denoted as $\mathcal{P}_1^{(q)}= \{(q-1)N_0+1,...,(q - \frac{1}{2})N_0\}$ and $\mathcal{P}_2^{(q)}= \{(q-\frac{1}{2})N_0+1,...,qT_0\}$, respectively. We denote $\mathbf{x}_{k,i}^{(q)}[n] \in \mathbb{C}^{M \times 1}$ as the transmit signal generated by the ISAC transceiver $k$ at time $n$ within $\mathcal{P}_i^{(q)}$. It is assumed that each PRI adopts an identical transmit waveform. The covariance matrix of the transmit signal, denoted as $\mathbf{Q}_{k,i}$, can be formulated as $\mathbf{Q}_{k,i} = \mathbb{E}[ \mathbf{x}_{k,i}^{(q)}[n] (\mathbf{x}_{k,i}^{(q)})^H[n] ]$, where $\mathbf{Q}_{k,i} \succeq 0$.
The full-duplex and half-duplex operations need to satisfy the following constraints, respectively:
\begin{subequations}
    \begin{align}
        \label{eqn:wide_mode_1}
        &\text{Full-duplex:} \quad \mathbf{Q}_{A,1} = \mathbf{Q}_{A,2}, \mathbf{Q}_{B,1} = \mathbf{Q}_{B,2},\\
        \label{eqn:wide_mode_2}
        &\text{Half-duplex:} \quad \mathbf{Q}_{A,2} = \mathbf{0}_{M \times M}, \mathbf{Q}_{B,2} = \mathbf{0}_{M \times M}.
    \end{align}
\end{subequations}
In the wideband system, the delay spread of the communication channels can often be substantially greater than the symbol period, leading to frequency-selective fading. Therefore, it is imperative to represent each communication channel and SI channel by multiple filter taps \cite{tse2005fundamentals}. Let us denote the $l$-th filter tap of the communication channel as $\mathbf{h}_{k, l} \in \mathbb{C}^{M \times 1},$ and similarly, the $l$-th filter tap of the SI channel as $\mathbf{g}_{k, l} \in \mathbb{C}^{M \times 1}$. We posit the existence of a total of $L$ filter taps within the communication channel, and correspondingly, $\tilde{L}$ filter taps within the SI channel. Consequently, the received signal at transceiver $k$ at time $n$ within $\mathcal{P}_m^{(q)}$ can expressed as:
\begin{align}
    y_{k,i}^{(q)}[n] = &\underbrace{\sqrt{\rho_c} u_{k,i}^{(q)}[n]}_{\text{communication signal}} + \underbrace{\sqrt{\rho_s} v_{k,i}^{(q)}[n]}_{\text{sensing echo signal}} \nonumber \\ &+ \underbrace{\sqrt{\eta} \sum_{l=0}^{\tilde{L}-1} \mathbf{g}_{k,l}^H \mathbf{x}_{k,i}^{(q)}[n- \tau_{\text{SI}} - l]}_{\text{SI}} + n_{k,i}^{(q)}[n],
\end{align}
where $\tau_{\text{SI}}$ denote the delay of the SI channel, and the signal $u_{k,i}^{(q)}[n]$ and $v_{k,i}^{(q)}[n]$ are given by, respectively,
\begin{align}
    &u_{k,i}^{(q)}[n] = \begin{cases}
        \sum_{l=0}^{L-1} \mathbf{h}_{k', l}^H \mathbf{x}_{k', 2}^{(q-1)}[n - \tau - l], &\text{if } m = 1, \\
        \sum_{l=0}^{L-1} \mathbf{h}_{k', l}^H \mathbf{x}_{k',1}^{(q)}[n - \tau - l], &\text{if } m = 2,
    \end{cases} \\
    &v_{k,i}^{(q)}[n] = \alpha_k e^{j 2 \pi \nu_k n T_s } \mathbf{a}^H(\theta_k) \mathbf{x}_{k,i}^{(q-1)}[n - 2\tau],
\end{align}
Then, the corresponding SI-canceled signal is given by
\begin{equation}
    \tilde{y}_{k,i}^{(q)} = \sqrt{\rho_c} u_{k,i}^{(q)}[n] + \sqrt{\rho_s} v_{k,i}^{(q)}[n] + e_{k,i}^{(q)}[n] + n_{k,i}^{(q)}[n],
\end{equation}
where $e_{k,i}^{(q)}[n]$ denotes the residual SI caused by limited receiver dynamic range and is modeled as $e_{k,i}^{(q)}[n] \sim \mathcal{CN}(0, \rho_{\text{SI}} \Phi_{k,i})$ with $\Phi_{k,i} = \eta \sum_{l=0}^{\tilde{L}-1} \mathbf{g}_{k,l}^H \mathbf{Q}_{k,i} \mathbf{g}_{k,l}$.

\subsection{Transmission and Reception Schemes}

Note that frequency-selective fading results in inter-symbol interference (ISI). To address this problem, the multi-carrier OFDM technique is typically applied in conventional wireless communication systems. However, the PAPR and inter-carrier interference caused by the Doppler frequency shift of the sensing target can lead to significant performance loss in OFDM-based sensing systems. To avoid these disadvantages, we exploit the recently proposed single-carrier DAM technique to address ISI through delay pre-compensation and path-based beamforming, which is more friendly for sensing \cite{lu2022delay, xiao2022integrated}. Following \cite{lu2022delay}, the DAM signal transmitted at transceiver $k$ at time $n$ within $\mathcal{P}_m^{(q)}$ is given by 
\begin{equation}
    \mathbf{x}_{k,i}^{(q)}[n] = \sum_{l=0}^{L-1} \mathbf{w}_{k,i,l} c_{k,i}^{(q)}[n- \kappa_l] + \mathbf{s}_{k,i}^{(q)}[n],
\end{equation}
where $\kappa_l = L - l$ denotes the delay pre-compensation for the $l$-th channel filter tap, $\mathbf{w}_{k,i,l} \in \mathbb{C}^{M \times  1}$ denotes the beamformer for the $l$-th channel filter tap, and $c_{k,i}^{(q)}[n] \sim \mathcal{CN}(0,1)$ denotes $n$-th unit-power independent information symbol. Therefore, the transmit covariance matrix can be rewritten as 
\begin{equation}
    \mathbf{Q}_{k,i} = \sum_{l=0}^{L-1} \mathbf{w}_{k,i,l} \mathbf{w}_{k,i,l}^H + \mathbf{R}_{k,i} \succeq \sum_{l=0}^{L-1} \mathbf{w}_{k,i,l} \mathbf{w}_{k,i,l}^H.
\end{equation}
Since the dedicated sensing signal $\mathbf{s}_{k,i}^{(q)}[n]$ can be \emph{a-priori} designed, it can be pre-canceled at the receiver. As such, after canceling the interference from the dedicated sensing signal, the communication signal $u_{k,1}^{(q)}[n]$ at transceiver $k$ with DAM becomes
\begin{align}
    &\tilde{u}_{k,1}^{(q)}[n] = \sqrt{\rho_c} \sum_{l=0}^{L-1} \sum_{l'=0}^{L-1} \mathbf{h}_{k',l}^H \mathbf{w}_{k',2,l'} c_{k',2}^{(q-1)}[n - \tau - \kappa_{l'}-l], \nonumber \\
    &\overset{(a)}{=} \sqrt{\rho_c} \left( \sum_{l=0}^{L-1} \mathbf{h}_{k',l}^H \mathbf{w}_{k',2,l} \right) c_{k',2}^{(q-1)}[n - \tau - L] \nonumber \\ &+ \sqrt{\rho_c} \underbrace{\sum_{l=0}^{L-1} \sum_{l'\neq l}^{L-1} \mathbf{h}_{k',l}^H \mathbf{w}_{k',2,l'} c_{k',2}^{(q-1)}[n-\tau-L + l'-l]}_{\text{inter-symbol interference}}.
\end{align}
In step $(a)$, all multi-path components are aligned in the first term with a common delay $\tau+L$ by exploiting the delay pre-compensation $\kappa_l$. Furthermore, the ISI collected in the second term can be suppressed by carefully designing the beamformers. If the zero-forcing (ZF) beamformers are exploited\footnote{To simplify the beamforming design and guarantee that the sensing signal is free from communication interference, we consider the ZF beamformers for eliminating the inter-symbol interference in DAM. Apart from the ZF beamformers, the maximal-ratio transmission (MRT) and minimum mean-square error (MMSE) beamformers can also be adopted \cite{lu2022delay}, which may exhibit better performance but result in high design complexity.}, i.e., $\mathbf{h}_{k,l}^H \mathbf{w}_{k,2,l'} = 0, \forall l \neq l',$ the signal $\tilde{u}_{k,1}^{(q)}[n]$ becomes ISI-free and can be simplified as follows:
\begin{equation}
    \tilde{u}_{k,1}^{(q)}[n] = \sqrt{\rho_c} \left( \sum_{l=0}^{L-1} \mathbf{h}_{k',l}^H \mathbf{w}_{k',2,l} \right) c_{k',2}^{(q-1)}[n-\tau- L].
\end{equation}
The ISI-free version of the signal $u_{k,2}^{(q)}[n]$ can be derived following the same path. Then, after decoding the communication signal, it can be removed from the received signal via SIC, resulting in the following effective sensing signal:
\begin{align}
    &\psi_{k,i}^{(q)}  = \sqrt{\rho_s} v_{k,i}^{(q)}[n] + z_{k,i}^{(q)}[n] \nonumber \\ & = \sqrt{\rho_s} \alpha_k e^{j 2 \pi \nu_k n T_s } \mathbf{a}^H(\theta_k) \mathbf{x}_{k,i}^{(q-1)}[n - 2\tau] + z_{k,i}^{(q)}[n],
\end{align} 
where $z_{k,i}^{(q)}[n] = e_{k,i}^{(q)}[n] + n_{k,i}^{(q)}[n]$.

\subsection{Performance Metrics}

\subsubsection{Achievable Rate of Communication}
The SINR for decoding information symbol from the received signal with DAM at transceiver $k$ at time $n \in \mathcal{P}_m^{(q)}$ is given by 
\begin{equation}
    \gamma_{k,i} = \frac{ \rho_c \left| \sum_{l=0}^{L-1} \mathbf{h}_{k',l}^H \mathbf{w}_{k',i',l} \right|^2}{\rho_s \mathbf{a}^H(\theta_k) \mathbf{Q}_{k,i} \mathbf{a}(\theta_k) + \rho_{\text{SI}} \Phi_{k,i} + 1},
\end{equation}
where $i' = 2$, if $m = 1$; and $i' = 1$, otherwise. The achievable communication rate at transceiver $k$ is thus given by 
\begin{equation}
    R_k = \frac{1}{2} \sum_{i=1}^2 \log_2 (1 + \gamma_{k,i}).
\end{equation}

\subsubsection{Cram{\'e}r-Rao Bound of Sensing}
To derive the CRB for estimating the directions, we first stack all the samples of the effective signals across the entire CPI, i.e., $U$ PRIs, into the vector $\boldsymbol{\psi}_k \in \mathbb{C}^{2N \times 1}$, which can be written as
\begin{equation}
    \boldsymbol{\psi}_k = \sqrt{\rho_s} \alpha_k (\mathbf{I}_{2N} \otimes \mathbf{a}^H(\theta_k)) \mathrm{vec}(\mathbf{X}_k \mathrm{diag}( \mathbf{d}(\nu_k) )) + \mathbf{z}_k,
\end{equation} 
where $\mathbf{X}_k \in \mathbb{C}^{M \times  2N}$ denotes the overall transmit signal matrix and $\mathbf{z}_k \in \mathbb{C}^{2N \times 1}$ denotes the corresponding interference plus noise vector. The expected CRB for estimating $\theta_k$ from $\boldsymbol{\psi}_k$ is derived in Appendix C as follows:
\begin{equation} \label{CRB_wideband}
    C_k = \frac{1}{2 \rho_s N} \left(\sum_{i =1}^2 \frac{ |\alpha_k|^2 \dot{\mathbf{a}}^H(\theta_k) \mathbf{Q}_{k,i} \dot{\mathbf{a}}(\theta_k)}{\rho_{\text{SI}} \Phi_{k,i} + 1} \right)^{-1}
\end{equation}

\subsection{Joint Beamforming Design for S\&C Tradeoff}

\subsubsection{Problem Formulation}
In this subsection, the path-based beamformers and the dedicated sensing signal are jointly optimized to characterize the tradeoff between the communication rate $R = \sum_{k =1}^2 R_k$ and the sensing CRB $C = \sum_{k =1}^2 C_k$. The corresponding scalarized SOOP by the weighted sum method is given by  
\begin{subequations} \label{problem:wide}
    \begin{align}
        \max_{ \mathbf{w}_{k, i, l}, \mathbf{Q}_{k,i} } & w R - (1-w) \mu C  \\
        \label{constraint:P2_1}
        \mathrm{s.t.} \quad & \eqref{eqn:wide_mode_1} \text{ or } \eqref{eqn:wide_mode_2}, \\
        \label{constraint:P2_2}
        & \sum_{m =1}^2 \mathrm{tr}(\mathbf{Q}_{k,i}) \le 2, \forall k, \\
        \label{constraint:P2_3}
        & \mathbf{h}_{k,l}^H \mathbf{w}_{k, i, l'} = 0, \forall l \neq l', \forall k, i,\\
        \label{constraint:P2_4}
        & \mathbf{Q}_{k,i} \succeq \sum_{l=0}^L \mathbf{w}_{k, i, l} \mathbf{w}_{k, i, l}^H, \forall k, i,
    \end{align}
\end{subequations}
where constraint \eqref{constraint:P2_1} determine whether the system works in full-duplex or half-duplex mode, constraint \eqref{constraint:P2_2} is the average power constraint, constraint \eqref{constraint:P2_3} is to eliminate the ISI caused by frequency-selective fading, and \eqref{constraint:P2_4} guarantees the existence of the dedicated sensing signal. In the sequel, we show that problem \eqref{problem:wide} can be solved by exploiting a similar method in Section \ref{sec:beamforming_narrow}.

\subsubsection{SCA-based Algorithm for Solving Problem \eqref{problem:wide}}
Firstly, to tractably recast non-convex constraint \eqref{constraint:P2_4}, we introduce the auxiliary variables $\{\mathbf{Q}_{k,i,l}\}$ and replace it with the following two constraints:
\begin{subequations} 
\begin{align}
    \label{constatint:P2_3_1}
    &\mathbf{Q}_{k,i} = \sum_{l=0}^L \mathbf{Q}_{k,i,l}, \forall k, i, \\
    \label{constatint:P2_3_0}
    &\mathbf{Q}_{k,i,l} \succeq \mathbf{w}_{k, i, l} \mathbf{w}_{k, i, l}^H, \forall k, i, l,
\end{align}
\end{subequations}
The constraint \eqref{constatint:P2_3_0} can be further transformed into the following convex form by exploiting the Schur complement condition:
\begin{subequations} 
    \begin{gather}
    \label{constatint:P2_3_2}
      \begin{bmatrix}
        \mathbf{Q}_{k,i,l} & \mathbf{w}_{k,i,l} \\
        \mathbf{w}_{k,i,l}^H & 1
      \end{bmatrix} \succeq 0, \forall k, i, l.
    \end{gather}
\end{subequations}
Then, following the same path in Section \ref{sec:beamforming_narrow}, problem \eqref{problem:wide} can be transformed into 
\begin{subequations} \label{problem:wide_2}
    \begin{align}
        \max_{\scriptstyle \mathbf{w}_{k, i, l}, \mathbf{Q}_{k,i,l}, \atop \scriptstyle  r_{k,i}, g_{k,i}, d_{k,i} } \quad & f(w, r_{k,i}, d_{k,i}) \\
        \mathrm{s.t.} \quad & \gamma_{k,i}^{[t]} \ge r_{k,i}, \forall k, i,\\
        &G_{k,i} \ge g_{k,i}^2, \forall k, i\\
        & D_{k,i}^{[t]} \ge d_{k,i}, \forall k, i\\    
        & \eqref{constraint:P2_1}-\eqref{constraint:P2_3}, \eqref{constatint:P2_3_1}, \eqref{constatint:P2_3_2}.
    \end{align}
\end{subequations}
Here, we have
\begin{align}
    D_{k,i}^{[t]} & \triangleq \frac{2 g_{k,i}^{[t]} }{\rho_{\text{SI}} \Phi_{k,i}^{[t]} + 1}g_{k,i} - \left( \frac{g_{k,i}^{[t]} }{\rho_{\text{SI}} \Phi_{k,i}^{[t]} + 1} \right)^2 (\rho_{\text{SI}} \Phi_{k,i} + 1), \\
    \gamma_{k,i}^{[t]} & \triangleq \frac{ 2 \rho_c \mathrm{Re} \left\{ \left( \sum_{l=0}^{l-1} \mathbf{w}_{k', i', l}^{[t],H} \mathbf{h}_{k', l} \right) \left( \sum_{l=0}^{l-1} \mathbf{h}_{k', l}^H \mathbf{w}_{k', i', l} \right)  \right\} }{I_{k,i}^{[t]}} \nonumber \\ &- \rho_c \left| \frac{\sum_{l=0}^{l-1} \mathbf{h}_{k', l}^H \mathbf{w}_{k', i', l}^{[t]}}{I_{k,i}^{[t]}} \right|^2 I_{k,i},
\end{align}
where $I_{k,i} = \rho_s \mathbf{a}^H(\theta_k) \mathbf{Q}_{k,i} \mathbf{a}(\theta_k) + \rho_{\text{SI}} \Phi_{k,i} + 1$. Problem \eqref{problem:wide_2} is convex and can be solved by the standard interior-point algorithm. The overall SCA-based algorithm for solving problem \eqref{problem:wide} is summarized in \textbf{Algorithm \ref{alg:SCA_wide}}.

\subsubsection{Convergence, Optimality, and Complexity Analysis} Following the analogous approach as presented in Section \ref{sec:alg_narrow}, we can readily establish both the convergence and the KKT optimality properties of \textbf{Algorithm \ref{alg:SCA_wide}}. 
For the sake of conciseness, we shall refrain from presenting the detailed proof here. Furthermore, we turn our attention to analyzing the worst-case computational complexity of \textbf{Algorithm \ref{alg:SCA_wide}}. Given that the problem \eqref{problem:wide_2} can be construed as a SDP problem, focusing on the semidefinite matrices as defined in \eqref{constatint:P2_3_2},
the worst-case computational complexity of \textbf{Algorithm \ref{alg:SCA_wide}} is given by $\mathcal{O} \left((4L)^3(N+1)^{3.5} + (4L)^4 \log(1/\epsilon)\right)$.

\begin{algorithm}[tb]
    \caption{SCA-based algorithm for solving problem \eqref{problem:wide}.}
    \label{alg:SCA_wide}
    \begin{algorithmic}[1]
        \STATE{Initialize feasible $\big\{\mathbf{w}_{k,i,l}^{[0]}, \mathbf{Q}_{k,i,l}^{[0]}, g_{k,i}^{[0]} \big\}$, and set $t=0$.}
        \REPEAT
        \STATE{update $\big\{\mathbf{w}_{k,i,l}^{[t]}, \mathbf{Q}_{k,i,l}^{[t]}, g_{k,i}^{[t]} \big\}$ by solving \eqref{problem:wide_2}}.
        \STATE{$t = t+1$.}
        \UNTIL{the fractional reduction of the objective value falls below a predefined threshold.}
    \end{algorithmic}
\end{algorithm}

\section{Numerical Results} \label{sec:results}
In this section, the numerical results obtained through Monte Carlo simulations are provided to compare the performance of full-duplex and half-duplex operations in both narrowband and wideband bidirectional ISAC systems. The adopted system parameters are presented in Table \ref{table:parameters}. All the convex problems in the paper are solved by the CVX toolbox \cite{cvx}. The scaling parameter is set to $\mu = 1.5 \times 10^4$. The initial points of the algorithms are set as random complex Gaussian variables that satisfy the corresponding constraints. Finally, all the simulation results are obtained by averaging over $100$ channel realizations unless otherwise specified.
\begin{table*}[tbp]
    \caption{System Parameters}
    \begin{center}
    \centering
    \resizebox{\textwidth}{!}{
        \begin{tabular}{|l|l|l||l|l|l|}
            \hline
            \centering
            $\Delta$ & Coherence time &$1$ ms & $N$ &Number of antennas at each ISAC transceiver &$10$\\
            \hline
            \centering
            $\rho_c$ &  SNR of the communication signal  & $15$ dB & $\rho_s$&  SNR of the sensing echo signal & $7$ dB\\
            \hline
            \centering
            $\rho_{\text{SI}}$& Limited dynamic range factor of the receiver  & $-80$ dB &$\eta$  &INR of the SI signal &$50$ dB \\
            \hline
            \centering
            $D$ & Distance between ISAC transceivers  & $300$ m   & $\theta_A$, $\theta_B$ &Angles of the ISAC transceivers &$0^{\circ}$\\
            \hline
        \end{tabular}
    }
    \end{center}
    \label{table:parameters}
\end{table*}

\subsection{Narrowband System}

For the narrowband bidirectional ISAC system, we set the bandwidth to $W = 100$ KHz. The flat-fading communication channels are assumed to obey the Rician fading, i.e.,
\begin{equation}
    \mathbf{h}_k = \sqrt{ \frac{\beta}{\beta + 1} } \mathbf{a}(\theta_k) + \sqrt{\frac{1}{\beta + 1}} \mathbf{h}_k^w, \forall k ,
\end{equation}
where $\beta$ is the Rician factor that characterizing the strength of the LOS component and $\mathbf{h}_k^w \in \mathbb{C}^{N \times 1}$ denotes the non-line-of-sight (NLOS) component that obeys Rayleigh fading, i.e., $\mathbf{h}_k^w \sim \mathcal{CN}(\mathbf{0}, \mathbf{I}_N)$. The SI channel $\mathbf{g}_k$ at each transceiver is assumed to obey Rayleigh fading, i.e., $\mathbf{g}_k \sim \mathcal{CN}(\mathbf{0}, \mathbf{I}_N)$. Moreover, according to the parameters in Table \ref{table:parameters}, the length of one CPI in the narrowband system is given by $2T = \Delta \times W = 100$.

\subsubsection{Convergence of Algorithm \ref{alg:SCA_narrow}}
In Fig. \ref{fig:convergence_narrow}, we present a demonstration of the convergence characteristics of \textbf{Algorithm \ref{alg:SCA_narrow}} when $w = 0.5$ and $\beta = 0$ dB. As illustrated in Fig. \ref{fig:convergence_narrow}, we observe that, in all instances, the objective value achieved by the proposed algorithm converges swiftly to a stable value within just a few iterations. This observation provides strong evidence in favor of the efficiency and low computational complexity of the proposed algorithm.

\subsubsection{S\&C Tradeoff}
In Fig. \ref{fig:tradeoff_narrow}, we delve into an exploration of the tradeoff between S\&C, while considering a range of Rician factors. It is notable and intriguing to observe that, unlike conventional communication-only systems, in the context of ISAC systems, \emph{full-duplex mode may not necessarily outperform half-duplex mode, even when SI has been effectively mitigated (indicated by $\rho_{\text{SI}} = -80$ dB).} 
This phenomenon is distinct from the outcomes typically encountered in communication-only scenarios. When the Rician factor is relatively small, such as in cases with $\beta=0$ dB and $\beta=6$ dB, full-duplex mode demonstrates superior performance in the communication-prior regime. However, in the sensing-prior regime, the half-duplex mode becomes the more favorable choice.
This behavior can be explained by several factors. Firstly, in the communication-prior regime, the sensing echo signals do not necessitate significant strength and can be effectively suppressed through beamforming. In this context, the interference incurred by sensing echo signals on the communication signals in full-duplex mode is minimal, leading to enhanced communication performance, as elucidated in \textbf{Remark \ref{remark_1}}. In the sensing-prior regime, guaranteeing the quality of sensing performance requires strong sensing echo signals, which, in turn, become potent sources of interference to the communication signals within the full-duplex mode. Therefore, communication performance in full-duplex mode is significantly reduced. These outcomes align with the discussions in \textbf{Remark \ref{remark_2}}.

\begin{figure}[t!]
    \centering
    \begin{subfigure}[t]{0.4\textwidth}
        \centering
        \includegraphics[width=1\textwidth]{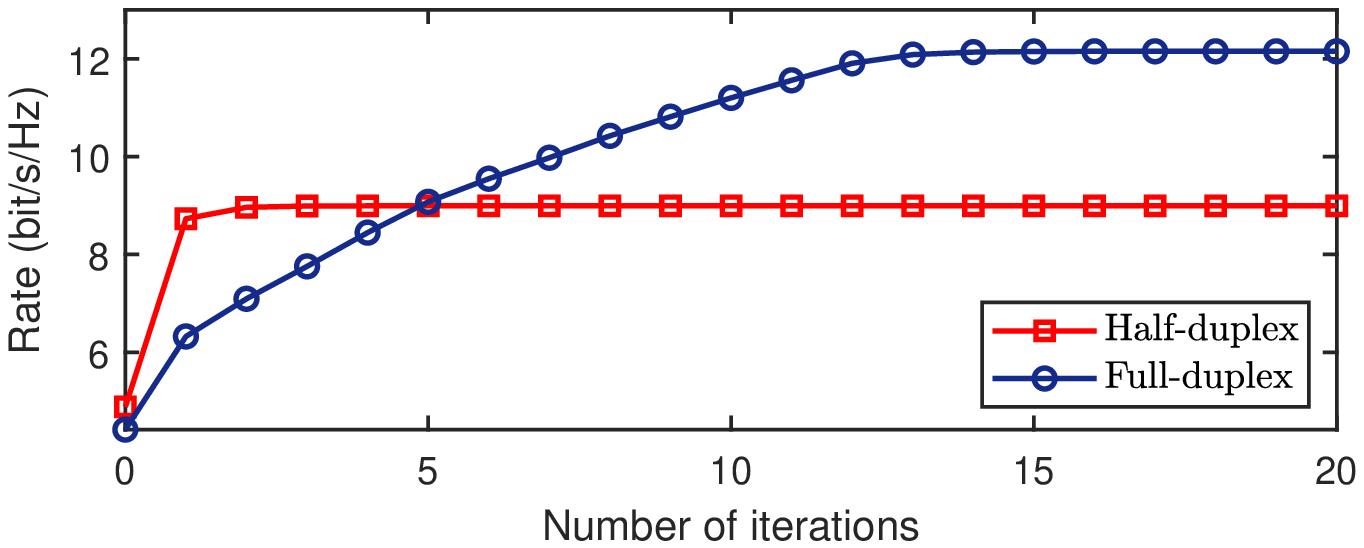}
    \end{subfigure}
    \begin{subfigure}[t]{0.4\textwidth}
        \centering
        \includegraphics[width=1\textwidth]{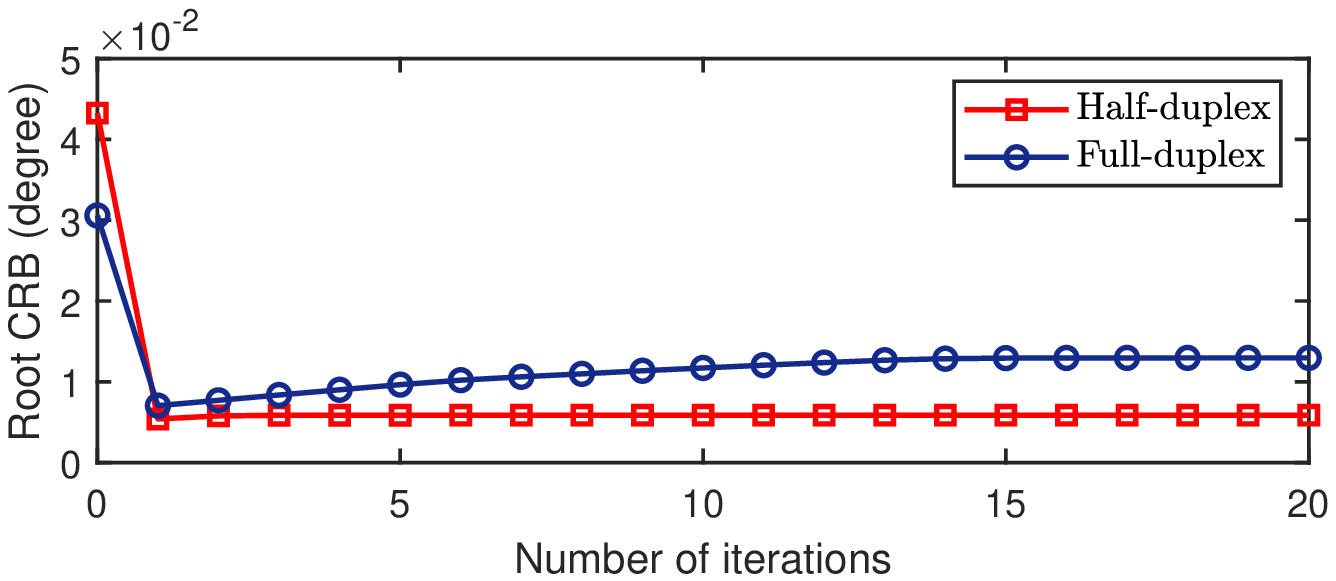}
    \end{subfigure}
    \caption{Convergence behavior of Algorithm \ref{alg:SCA_narrow} when $w = 0.5$. }
    \label{fig:convergence_narrow}
\end{figure}

\begin{figure}[t!]
    \centering
    \includegraphics[width=0.4\textwidth]{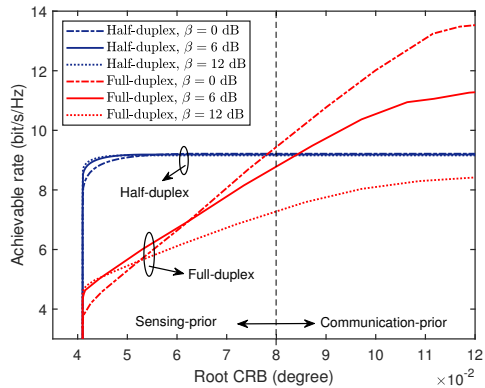}
    \caption{S\&C Tradeoff region achieved by full-duplex and half-duplex operations in the narrowband system.}
    \label{fig:tradeoff_narrow}
\end{figure}

It is also discernible from Fig. \ref{fig:tradeoff_narrow} that the tradeoff region achieved by the full-duplex mode tends to contract as the Rician factor increases, whereas the tradeoff region obtained through the half-duplex mode experiences a slight expansion. This phenomenon can be attributed to the fact that a higher Rician factor implies a stronger correlation between the communication and sensing channels. In full-duplex mode, a larger Rician factor results in a scenario where the objective of maximizing the power of the communication signal and minimizing the power of the sensing echo signal becomes increasingly contradictory. Consequently, the tradeoff region diminishes. In half-duplex mode, there is no interference stemming from the sensing echo signal affecting the communication signal. Therefore, for large Rician factors, maximizing the power of the sensing echo signal at one ISAC transceiver aligns with the objective of maximizing the power of the communication signal at the other ISAC transceiver, leading to a broader tradeoff region. In the following, we delve further into additional numerical results to thoroughly examine the influence of the Rician factor on the S\&C performance.

\begin{figure}[t!]
    \centering
    \includegraphics[width=0.4\textwidth]{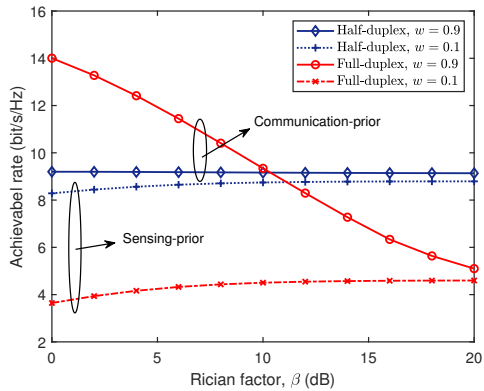}
    \caption{Achievable Rate versus Rician Factor $\beta$.}
    \label{fig:rician_rate_narrow}
\end{figure}

\begin{figure}[t!]
    \centering
    \includegraphics[width=0.4\textwidth]{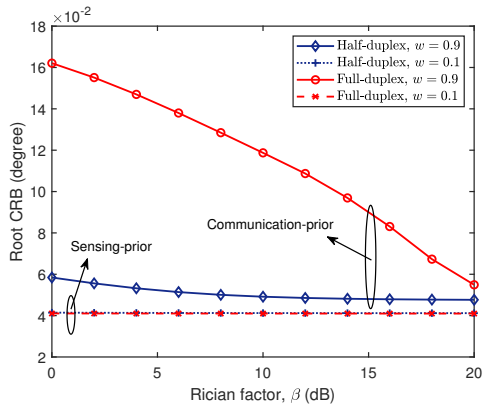}
    \caption{Root CRB versus Rician Factor $\beta$.}
    \label{fig:rician_crb_narrow}
\end{figure}

\subsubsection{Imapct of Rician Factor}
Fig. \ref{fig:rician_rate_narrow} provides insight into the influence of the Rician factor on the achievable rates in both the communication-prior regime ($w = 0.9$) and the sensing-prior regime ($w = 0.1$). It is evident that in the sensing-prior regime, the half-duplex mode consistently outperforms the full-duplex mode in terms of achievable rates, irrespective of the Rician factor. This is primarily attributed to the absence of interference from the sensing echo signal to the communication signal. Furthermore, in the sensing-prior regime, the achievable rate experiences a slight increase as the Rician factor grows, owing to the higher correlation between the communication and sensing channels. Conversely, in the communication-prior regime, due to the interference stemming from the sensing echo signal, the performance advantage in achievable rates achieved by the full-duplex mode gradually diminishes, eventually falling below that of the half-duplex mode as the Rician factor increases. Nevertheless, the performance of the half-duplex mode remains relatively unaffected by variations in the Rician factor.

In Fig. \ref{fig:rician_crb_narrow}, we examine the impact of the Rician factor on the root CRB. In the sensing-prior regime, the root CRB achieved by both duplex modes remains stable and unaltered. We also note that the performance gap between the two protocols is negligible, aligning with the analysis presented in \textbf{Remark \ref{remark_1}}. However, in the communication-prior regime, the half-duplex mode emerges as the more favorable choice, and a noticeable performance gap between the two protocols becomes evident. This is due to the necessity of suppressing the sensing echo signal to ensure the communication performance within the full-duplex mode. Furthermore, this performance gap gradually diminishes as the Rician factor increases.

\begin{figure}[t!]
    \centering
    \begin{subfigure}[t]{0.4\textwidth}
        \centering
        \includegraphics[width=1\textwidth]{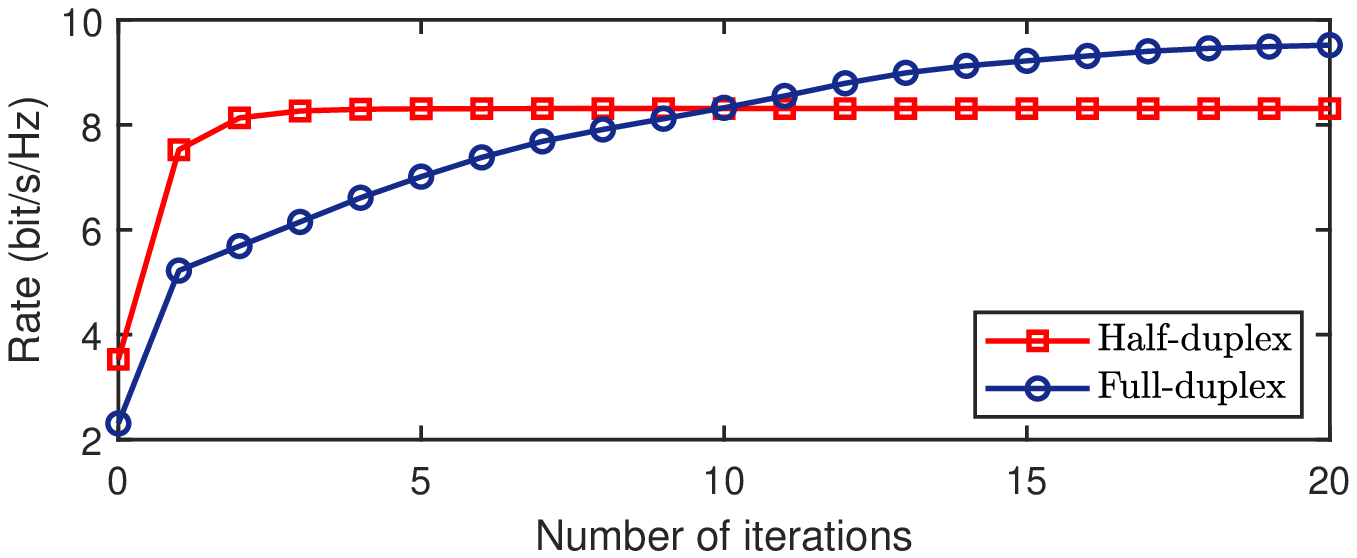}
    \end{subfigure}
    \begin{subfigure}[t]{0.4\textwidth}
        \centering
        \includegraphics[width=1\textwidth]{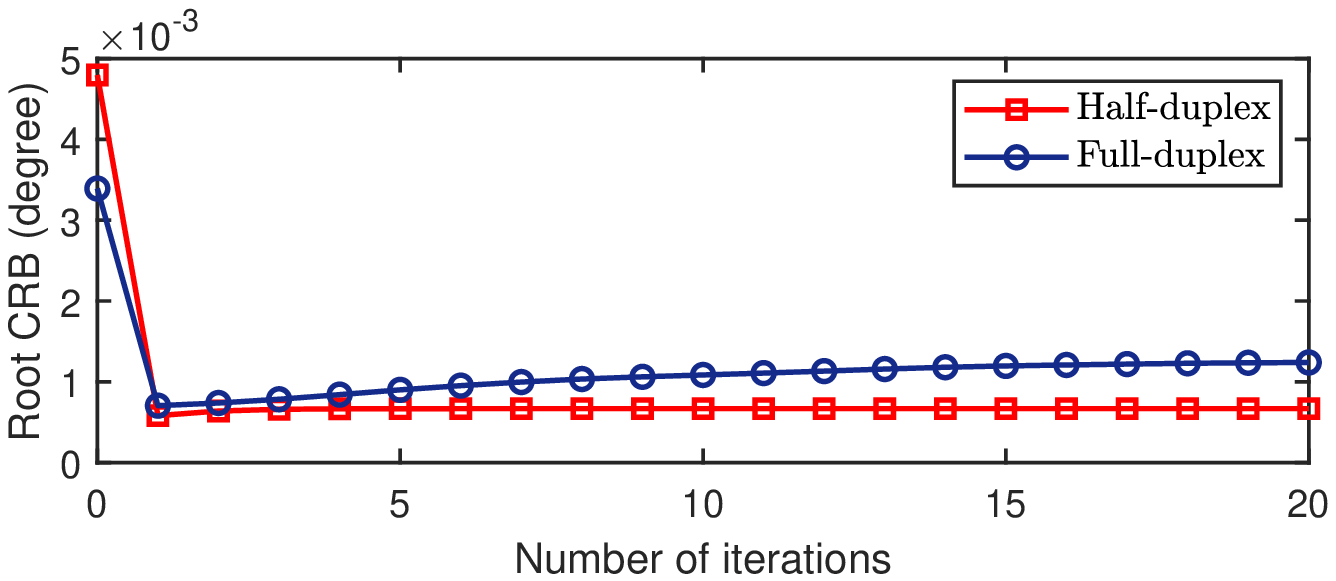}
    \end{subfigure}
    \caption{Convergence behavior of Algorithm \ref{alg:SCA_wide} when $w = 0.5$. }
    \label{fig:convergence_braod}
\end{figure}

\begin{figure}[t!]
    \centering
    \includegraphics[width=0.4\textwidth]{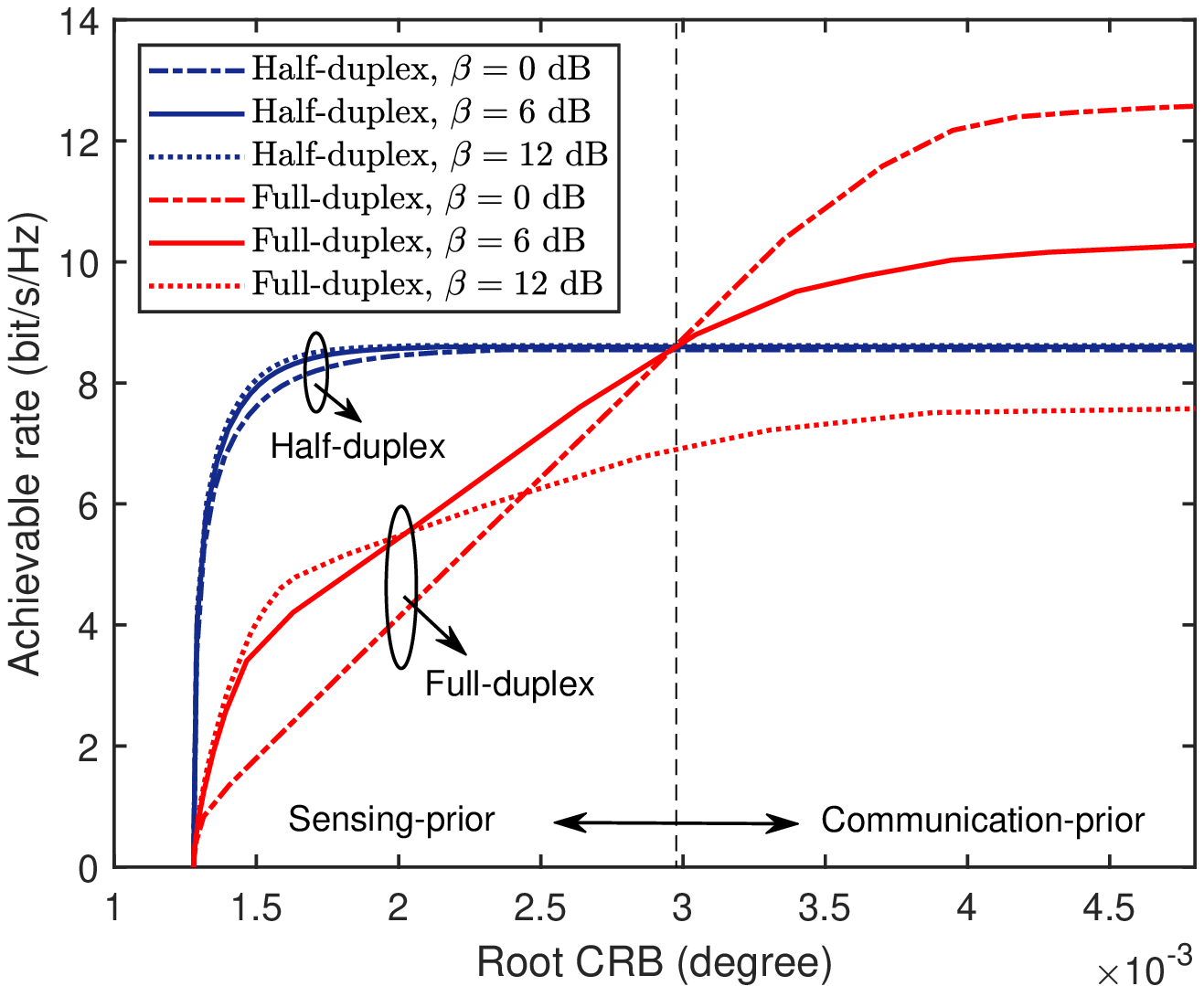}
    \caption{S\&C Tradeoff region achieved by full-duplex and half-duplex operations in the wideband system.}
    \label{fig:tradeoff_broad}
\end{figure}

\begin{figure*}[t!]
    \centering
    \begin{subfigure}[t]{0.4\textwidth}
        \centering
        \includegraphics[width=1\textwidth]{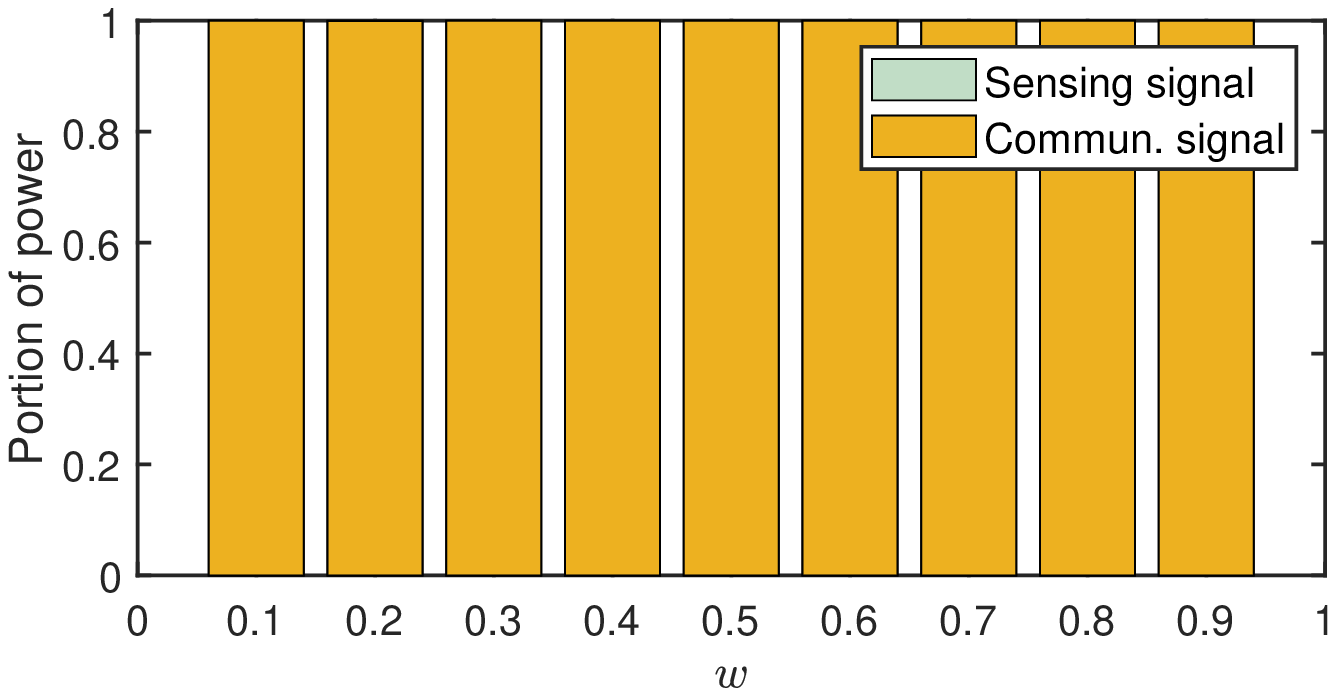}
        \caption{Narrowband, Half-duplex}
    \end{subfigure}
    \begin{subfigure}[t]{0.4\textwidth}
        \centering
        \includegraphics[width=1\textwidth]{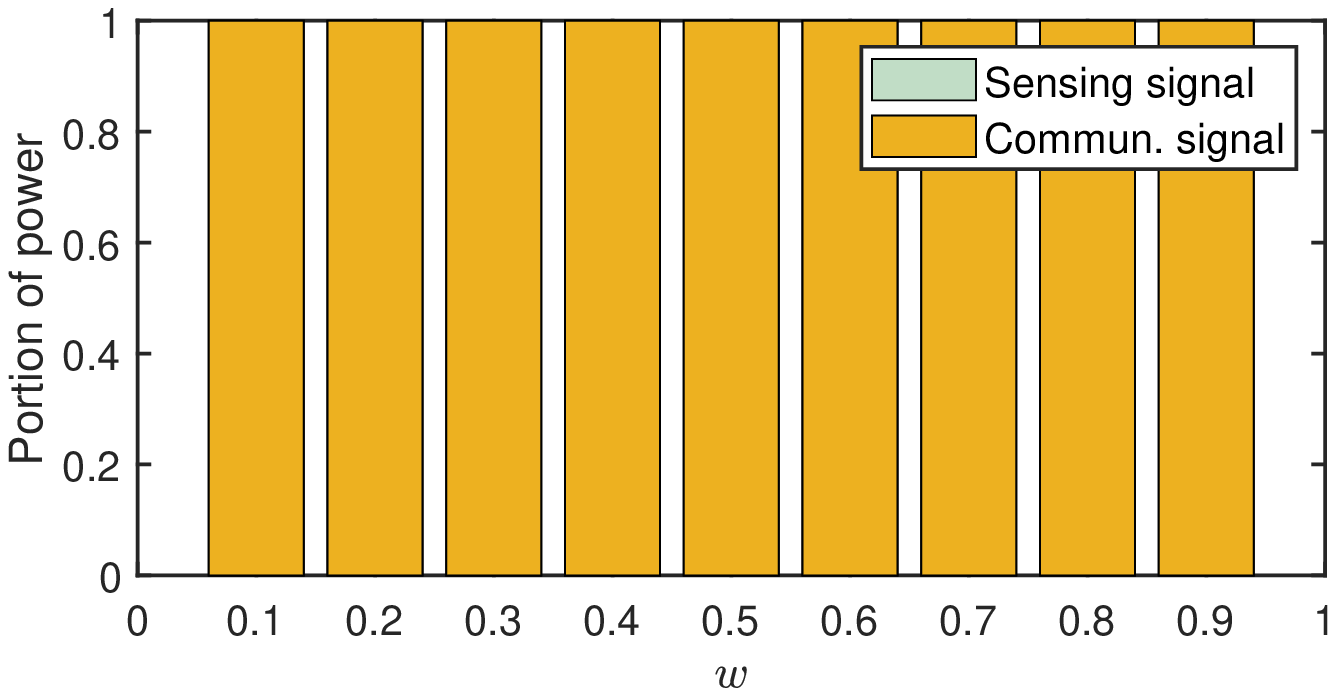}
        \caption{Narrowband, Full-duplex}
    \end{subfigure}
    \begin{subfigure}[t]{0.4\textwidth}
        \centering
        \includegraphics[width=1\textwidth]{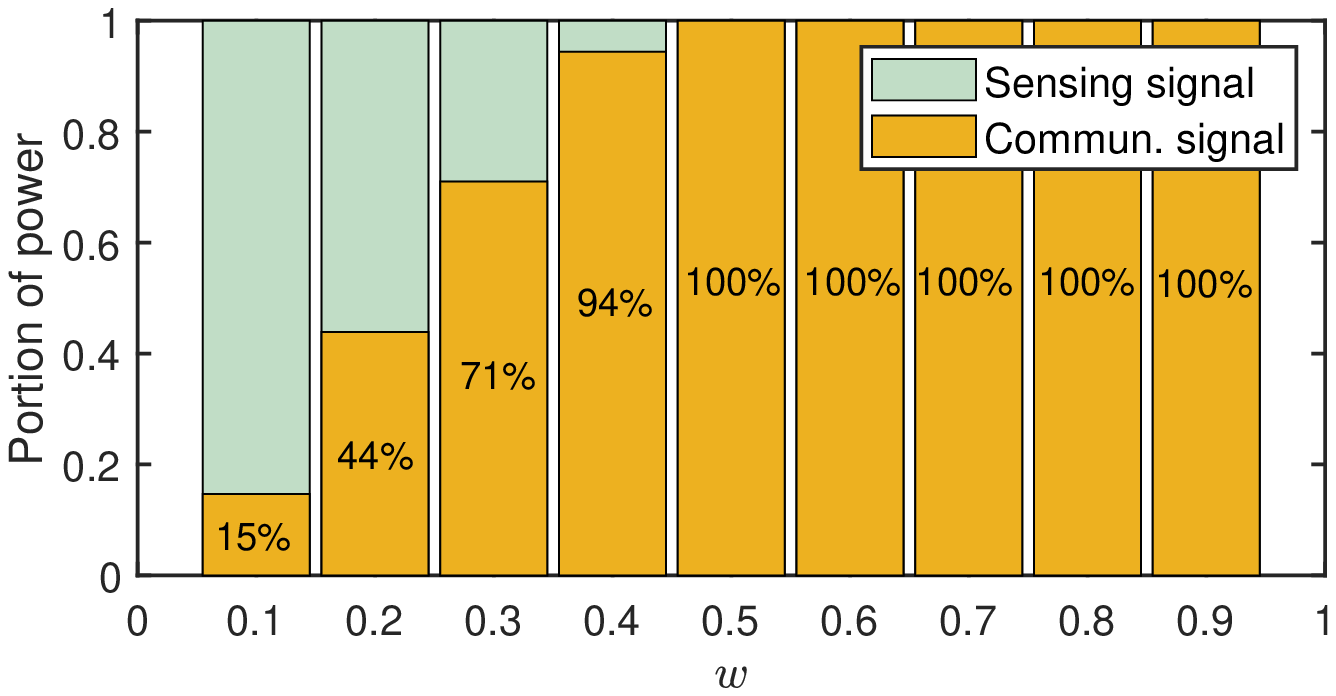}
        \caption{Wideband, Half-duplex}
    \end{subfigure}
    \begin{subfigure}[t]{0.4\textwidth}
        \centering
        \includegraphics[width=1\textwidth]{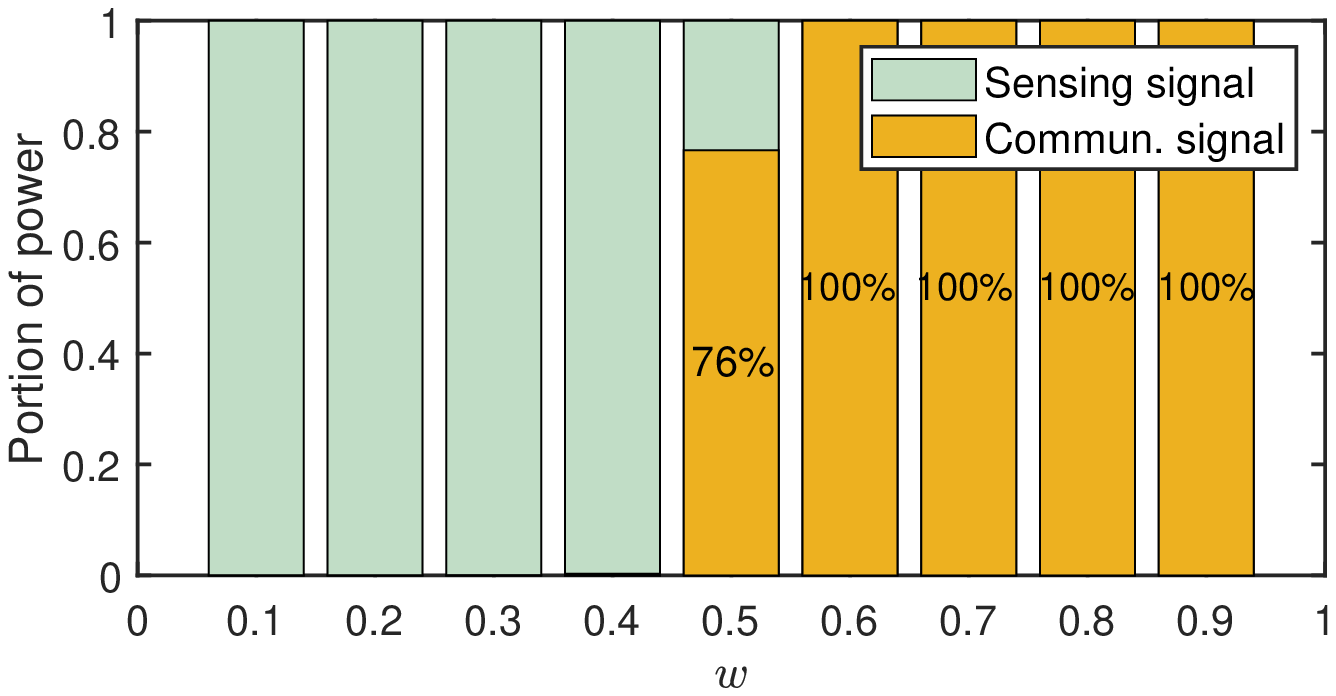}
        \caption{Wideband, Full-duplex}
    \end{subfigure}
    \caption{Power allocation of the communication and dedicated sensing signals.}
    \label{fig:power_allocation}
\end{figure*}

\subsection{Wideband System}
For the wideband bidirectional ISAC system, the system bandwidth is designated as $W = 100$ MHz, and the number of filter taps in each communication channel and SI channel is set to $L = \tilde{L} = 4$. According to \cite{tse2005fundamentals, mckay2006capacity}, we represent the first filter tap ($l = 0$) as a Rician fading channel, described by the equation:
\begin{equation}
    \mathbf{h}_{k,0} = \sqrt{ \frac{ \beta_0  }{ \beta_0 + 1} } \sigma_0 \mathbf{a}(\theta_k) + \sqrt{\frac{1}{ \beta_0 + 1}} \mathbf{h}_{k,0}^w, \forall k.
\end{equation}
Here, $\mathbf{h}_{k,0}^w \sim \mathcal{CN}(\mathbf{0}, \sigma_0^2 \mathbf{I}_M)$, where $\sigma_0$ denotes the power delay profile. The remaining filter taps are modeled as Rayleigh fading channels, denoted as $\mathbf{h}_{k,l} \sim \mathcal{CN}(\mathbf{0}, \sigma_l^2 \mathbf{I}_M) \forall  l \neq 0$, with $\sigma_l$ representing their respective power delay profiles. Additionally, these power delay profiles satisfy the constraint $\sum_{l=0}^{L-1} \sigma_l^2 = 1$ \cite{mckay2006capacity}.
In a similar fashion, the filter taps of the SI channels are modeled as Rayleigh fading channels and adhere to the same power delay profile constraint. To quantify the strength of the LOS component relative to the NLOS components distributed across all filter taps, we define the modified Rician factor $\beta$ as:
\begin{equation}
    \beta = \frac{\beta_0 \sigma_0^2 }{(\beta_0 + 1) \sum_{l=1}^{L-1} \sigma_l^2 + \sigma_0^2 }.
\end{equation}
Furthermore, in accordance with the parameters presented in Table \ref{table:parameters}, the delay of the communication signals is determined as $\tau = D/c \times W = 100$. The lengths of the PRI and CPI are thereby givn by $N_0 = 2\tau = 200$ and $2N = \Delta \times B = 10^5$, respectively.

\subsubsection{Convergence of Algorithm \ref{alg:SCA_wide}}
In Fig. \ref{fig:convergence_braod}, we investigated the convergence characteristics of the proposed \textbf{Algorithm \ref{alg:SCA_wide}} when $w = 0.5$ and $\beta = 0$ dB.
It can be observed that \textbf{Algorithm \ref{alg:SCA_wide}} exhibits a convergence behavior analogous to that of \textbf{Algorithm \ref{alg:SCA_narrow}}, with the objective value converging to a stable value within a small number of iterations.

\subsubsection{S\&C Tradeoff}
In Fig. \ref{fig:tradeoff_broad}, we present the S\&C tradeoff region achieved through the utilization of the DAM-based wideband transmission. It is discernible that the S\&C tradeoff regions attained by full-duplex and half-duplex mode exhibit a similar shape to that observed in the narrowband system.
To elaborate further, when the Rician factor is low, the half-duplex mode excels in the sensing-prior regime, while the full-duplex mode outperforms in the communication-prior regime. However, in scenarios with a very high Rician factor, the half-duplex mode emerges as the superior choice in both regimes.

Furthermore, through a comparative analysis between Fig. \ref{fig:tradeoff_narrow} and Fig. \ref{fig:tradeoff_broad}, we observe that the S\&C tradeoff in the wideband system is more prominent than in the narrowband system. In other words, in the wideband system, the performance of one function experiences a more substantial decline as the performance of the other function improves. This observation aligns with expectations, as the S\&C tradeoff in the narrowband system is primarily driven by interference between the two functions, whereas in the wideband system, it is additionally influenced by power allocation between communication signals and dedicated sensing signals. In subsequent simulations, we will further substantiate and demonstrate this phenomenon.

\subsubsection{Power Allocation}

In Fig. \ref{fig:power_allocation}, we study the power allocation between communication and dedicated sensing signals in both narrowband and wideband systems. Notably, in the narrowband system, there is a consistent absence of power allocated to dedicated sensing signals. This observation signifies that the communication signals are capable of fully providing the DoFs required for both functions, rendering the dedicated sensing signal superfluous. Conversely, in the wideband system, it is noteworthy to observe that the dedicated sensing signal consumes a substantial portion of the allocated power, particularly in the sensing-prior regime. This phenomenon can be elucidated as follows. In the wideband system, the optimization of communication signals is bound by zero-forcing constraints to ensure the absence of ISI. These zero-forcing constraints significantly curtail the available DoFs of the communication signals for the sensing function. As a result, to approach optimal sensing performance within the sensing-prior regime, a larger allocation of power to the dedicated sensing signal is required. Given that the dedicated sensing signal makes no contribution to communication, this increased allocation of power can lead to a dramatic reduction in communication performance. Therefore, as previously discussed, the S\&C tradeoff in the wideband system becomes more pronounced than that in the narrowband system.

\section{Conclusion} \label{sec:conslution}

The proposed bidirectional ISAC systems, comprising both full-duplex and half-duplex operation modes, were presented. Dedicated transmission and reception schemes were devised for both narrowband and wideband systems. Leveraging these schemes, the S\&C tradeoff was meticulously characterized through the optimization of joint S\&C signals. The numerical results have uncovered an intriguing observation. It has been revealed that, for both narrowband and wideband systems, \emph{full-duplex mode may not consistently outperform half-duplex mode within the bidirectional ISAC system.} Specifically, the full-duplex mode demonstrates superior performance in the communication-prior regime when the communication channels are not predominantly LOS, while the half-duplex mode is favored in other scenarios. As a result, practical bidirectional ISAC system design necessitates the adoption of a hybrid approach, capable of dynamic mode switching between full-duplex and half-duplex mode based on the specific operational context. Furthermore, it is pertinent to mention that this study delved into the direction-sensing task under specific assumptions. Nevertheless, the findings and conclusions presented in this paper hold potential relevance and applicability to other sensing tasks within more generalized scenarios. This assertion stems from the fundamental fact that the preservation of the sensing echo signal essentially transforms it into “unsuppressed SI” for communication systems. Hence, a comprehensive understanding of the full-duplex ISAC system necessitates further research efforts.

\vspace{-0.2cm}
\begin{appendices}

\section{Derivation of CRB in Narrowband Systems} 
The CRB can be calculated as the inverse of the Fisher information.
Define $\mathbf{u}_k = \sqrt{\rho_s} \alpha_k (\mathbf{I}_{2N} \otimes \mathbf{a}^H(\theta_k)) \mathrm{vec} (\mathbf{X}_k \mathrm{diag}( \mathbf{d}(\nu_k) ) )$.  Following the same path in \cite{kay1993fundamentals}, when other target parameters are known, the expected Fisher information for estimating $\theta_k$ from $\boldsymbol{\psi}_k$ is given by 
\begin{equation} \label{eqn:FIM_narrow} 
  J_k= \mathbb{E} \left[ 2 \mathrm{Re}\left\{ \frac{\partial \mathbf{u}_k^H}{\partial \theta_k} \mathbf{R}_{\mathbf{z}_k}^{-1} \frac{\partial \mathbf{u}_k}{\partial \theta_k}\right\} \right],
\end{equation}
where $\mathbf{R}_{\mathbf{z}_k}$ denotes the covariance matrix of $\mathbf{z}_k$ and the partial derivation is given by:
\begin{equation}
  \frac{\partial \mathbf{u}_k}{\partial \theta_k} = \alpha_k (\mathbf{I}_{2N} \otimes \dot{\mathbf{a}}^H(\theta_k)) \mathrm{vec}(\mathbf{X}_k \mathrm{diag}( \mathbf{d}(\nu_k) )),
\end{equation}
with $\dot{\mathbf{a}}(\theta_k)$ denoting the partial derivation of $\mathbf{a}(\theta_k)$ with respect to $\theta_k$ as follows:
\begin{align} \label{eqn:derivation}
  &\dot{\mathbf{a}}(\theta_k) = \mathrm{diag}(\mathbf{a}(\theta_k)) \nonumber \\
  &\times \left[0, -j \frac{2\pi}{\lambda}d \cos\theta_k,...,-j \frac{2\pi}{\lambda} (M-1) d \cos\theta_k \right]^T.
\end{align}  
To derive explicit expression of $J_k$, we first calculate the covariance matrices $\mathbf{R}_{\mathbf{X}_k} = \mathbb{E}[\mathrm{vec}(\mathbf{X}_k) \mathrm{vec}(\mathbf{X}_k)^H]$ and $\mathbf{R}_{\mathbf{z}_k} = \mathbb{E}[\mathbf{z}_k \mathbf{z}_k^H]$. Note that the cross-correlations of the signal $\mathbf{x}_{k,i}[n]$ and the noise $z_{k,i}[n]$ are equal to zero, i.e., $\mathbb{E}[\mathbf{x}_{k,i}[n] \mathbf{x}_{k, i'}^H[n'] ] = \mathbf{0}_{M \times M}$ and $\mathbb{E}[ z_{k,i}[n] z_{k,i'}^*[n'] ] = 0$ if $n \neq n'$ or $i \neq i'$. Therefore, the covariance matrices $\mathbf{R}_{\mathbf{X}_k}$ and $\mathbf{R}_{\mathbf{z}_k}$ can be calculated as follows:
\begin{align}
  &\mathbf{R}_{\mathbf{X}_k} = \begin{bmatrix}
    \mathbf{I}_{N} \otimes \mathbf{Q}_{k,1} & \mathbf{0} \\
    \mathbf{0} & \mathbf{I}_{N} \otimes \mathbf{Q}_{k,2}
  \end{bmatrix}, \\
  &\mathbf{R}_{\mathbf{z}_k} = \begin{bmatrix}
    (\rho_{\text{SI}} \Phi_{k,1} + 1)\mathbf{I}_{N} & \mathbf{0} \\
    \mathbf{0} & (\rho_{\text{SI}} \Phi_{k,2} + 1) \mathbf{I}_{N}
  \end{bmatrix}.
\end{align}
Based on the expression of $\mathbf{d}(\nu_k)$, it can be readily shown that $\mathbb{E}[ \mathrm{vec}(\mathbf{X}_k \mathrm{diag}( \mathbf{d}(\nu_k) )) \mathrm{vec}(\mathbf{X}_k \mathrm{diag}( \mathbf{d}(\nu_k) ))^H ] = \mathbf{R}_{\mathbf{X}_k}$. 
As such, the explicit expression of $J_k$ is given by
\begin{align}
  J_k = & \mathbb{E} [ 2 |\alpha_k|^2 \mathrm{Re} \{ \rho_s \mathrm{vec}(\mathbf{X}_k \mathrm{diag}( \mathbf{d}(\nu_k) ))^H (\mathbf{I}_{2N} \otimes \dot{\mathbf{a}}(\theta_k)) \mathbf{R}_{\mathbf{z}_k}^{-1} \nonumber 
  \\ & \times (\mathbf{I}_{2N} \otimes \dot{\mathbf{a}}^H(\theta_k)) \mathrm{vec}(\mathbf{X}_k \mathrm{diag}( \mathbf{d}(\nu_k) )) \} ] \nonumber \\
  = &2 |\alpha_k|^2 \rho_s \mathrm{tr} ( (\mathbf{I}_{2N} \otimes \dot{\mathbf{a}}(\theta_k)) \mathbf{R}_{\mathbf{z}_k}^{-1} (\mathbf{I}_{2N} \otimes \dot{\mathbf{a}}^H(\theta_k)) \nonumber \\
  &  \times \mathbb{E}[ \mathrm{vec}(\mathbf{X}_k \mathrm{diag}( \mathbf{d}(\nu_k) )) \mathrm{vec}(\mathbf{X}_k \mathrm{diag}( \mathbf{d}(\nu_k) ))^H ] ) \nonumber \\  
  = &2 |\alpha_k|^2 \rho_s \mathrm{tr} \left( (\mathbf{I}_{2N} \otimes \dot{\mathbf{a}}(\theta_k)) \mathbf{R}_{\mathbf{z}_k}^{-1} (\mathbf{I}_{2N} \otimes \dot{\mathbf{a}}^H(\theta_k)) \mathbf{R}_{\mathbf{X}_k} \right) \nonumber \\
  = &2 |\alpha_k|^2 \rho_s N \sum_{i =1}^2 \frac{\dot{\mathbf{a}}^H(\theta_k) \mathbf{Q}_{k,i} \dot{\mathbf{a}}(\theta_k)}{\rho_{\text{SI}} \Phi_{k,i} + 1}.
\end{align}
Then, the expected CRB is given by $J_k^{-1}$, resulting in the expression in \eqref{eqn:CRB}.

\begin{figure*}[!t]
  \normalsize
  \setcounter{equation}{54}
  \begin{subequations}
    \begin{gather}
      \label{KKT_1}
      \nabla_{\boldsymbol{\chi}} \mathcal{L} \big|_{\boldsymbol{\chi} = \boldsymbol{\chi}^{[\infty]}} = -\nabla_{\boldsymbol{\chi}} f(w, r_{k,i}, d_{k,i}) \big|_{\boldsymbol{\chi} = \boldsymbol{\chi}^{[\infty]}} - \sum_{k =1}^2 \sum_{i =1}^2 \lambda_{1, k,i} \nabla_{\boldsymbol{\chi}} ( r_{k,i} - \gamma_{k,i}^{[\infty]} ) \big|_{\boldsymbol{\chi} = \boldsymbol{\chi}^{[\infty]}} \nonumber \\
      - \sum_{k =1}^2 \sum_{i =1}^2 \lambda_{2, k,i} \nabla_{\boldsymbol{\chi}}( d_{k,i} - D_{k,i}^{[\infty]} ) \big|_{\boldsymbol{\chi} = \boldsymbol{\chi}^{[\infty]}} + \nabla_{\boldsymbol{\chi}} \mathcal{C} \big|_{\boldsymbol{\chi} = \boldsymbol{\chi}^{[\infty]}} = \mathbf{0}, \\
      \lambda_{1, k,i}( r_{k,i} - \gamma_{k,i}^{[\infty]} ) \big|_{\boldsymbol{\chi} = \boldsymbol{\chi}^{[\infty]}} = 0, \forall k, i, \\
      \lambda_{2, k,i}( d_{k,i} - D_{k,i}^{[\infty]} ) \big|_{\boldsymbol{\chi} = \boldsymbol{\chi}^{[\infty]}} = 0, \forall k, i, \\
      \lambda_{3,k} \left( \sum_{i =1}^2 \mathrm{tr}( \mathbf{Q}_{k,i} ) - 2 \right) \big|_{\boldsymbol{\chi} = \boldsymbol{\chi}^{[\infty]}} = 0, \forall k, \\
      \label{KKT_2}
      \mathrm{tr} \left( \mathbf{\Upsilon}_{1,k,i} ( \mathbf{Q}_{k,i} - \mathbf{w}_{k,i}(\mathbf{w}_{k,i})^H ) \right) \big|_{\boldsymbol{\chi} = \boldsymbol{\chi}^{[\infty]}} = 0, \forall k, i.
    \end{gather}
  \end{subequations}
  \hrulefill
  \vspace*{4pt}
\end{figure*}

\section{Proof of Proposition \ref{proposition_1}} 
Denote $\boldsymbol{\chi} = \{\mathbf{w}_{k,i}, \mathbf{Q}_{k,i}, r_{k,i}, \varpi_{k,i}, d_{k,i} \}$ as all the optimization variables of problem \eqref{problem:narrow_3} and $\boldsymbol{\chi}^{[t]} = \{ \mathbf{w}_{k,i}^{[t]}, \mathbf{Q}_{k,i}^{[t]}, r_{k,i}^{[t]}, \varpi_{k,i}^{[t]}, d_{k,i}^{[t]} \}$ as the solution obtained at the $t$-th iteration of \textbf{Algorithm \ref{alg:SCA_narrow}}. Due to the equivalence between problem \eqref{problem:narrow} and problem \eqref{problem:narrow_1}, \textbf{Proposition \ref{proposition_1}} holds if the converged solution $\boldsymbol{\chi}^{[\infty]}$ of problem \eqref{problem:narrow_3} as $t \rightarrow \infty$ satisfies the KKT conditions of problem \eqref{problem:narrow_1}. To prove this, we first give the following results:
\begin{subequations}
  \begin{gather}
    \label{eqn:proof_equality_1}
    \gamma_{k,i} \big|_{\boldsymbol{\chi} = \boldsymbol{\chi}^{[\infty]}} = \gamma_{k,i}^{[\infty]} \big|_{\boldsymbol{\chi} = \boldsymbol{\chi}^{[\infty]}}, \\
    \label{eqn:proof_equality_1.1}
    D_{k,i} \big|_{\boldsymbol{\chi} = \boldsymbol{\chi}^{[\infty]}} = D_{k,i}^{[\infty]} \big|_{\boldsymbol{\chi} = \boldsymbol{\chi}^{[\infty]}}, \\
    \label{eqn:proof_equality_2}
    \nabla_{\boldsymbol{\chi}} \gamma_{k,i} \big|_{\boldsymbol{\chi} = \boldsymbol{\chi}^{[\infty]}} = \nabla_{\boldsymbol{\chi}} \gamma_{k,i}^{[\infty]} \big|_{\boldsymbol{\chi} = \boldsymbol{\chi}^{[\infty]}}, \\
    \label{eqn:proof_equality_3}
    \nabla_{\boldsymbol{\chi}} D_{k,i} \big|_{\boldsymbol{\chi} = \boldsymbol{\chi}^{[\infty]}} = \nabla_{\boldsymbol{\chi}} D_{k,i}^{[\infty]} \big|_{\boldsymbol{\chi} = \boldsymbol{\chi}^{[\infty]}}. 
  \end{gather}
\end{subequations}
Specifically, \eqref{eqn:proof_equality_1} and \eqref{eqn:proof_equality_1.1} can be directly obatined by substituting $\boldsymbol{\chi} = \boldsymbol{\chi}^{[\infty]}$ into $\gamma_{k,i}$, $\gamma_{k,i}^{[\infty]}$, $D_{k,i}$, and $D_{k,i}^{[\infty]}$, respectively. Since $\gamma_{k,i}$ and $\gamma_{k,i}^{[\infty]}$ are only related to the optimization variables $\mathbf{w}_{k',m}$ and $\mathbf{Q}_{k,i}$, \eqref{eqn:proof_equality_2} can be obtained based on the following equations:
\begin{subequations}
  \begin{align}
    \nabla_{\mathbf{w}_{k',m}} & \gamma_{k,i}^{[\infty]} \big|_{\mathbf{w}_{k',m} = \mathbf{w}_{k',m}^{[\infty]}} \nonumber \\ 
    =&    \frac{2 \rho_c \mathrm{Re} \left\{ \mathbf{h}_{k'}^H \mathbf{w}_{k',m}^{[\infty]}\mathbf{h}_{k'}  \right\}}{ I_{k,i}^{[\infty]} }    = \nabla_{\mathbf{w}_{k',m}} \gamma_{k,i} \big|_{\mathbf{w}_{k',m} = \mathbf{w}_{k',m}^{[\infty]}}, \\
    \nabla_{\mathbf{Q}_{k,i}} &\gamma_{k,i}^{[\infty]} \big|_{\mathbf{Q}_{k,i}  = \mathbf{Q}_{k,i}^{[\infty]}} \nonumber \\ \
    = &    -\rho_c \left| \frac{\mathbf{h}_{k'}^H \mathbf{w}_{k',m}^{[\infty]}}{I_{k,i}^{[\infty]}} \right|^2 (\rho_s \mathbf{a}(\theta_k) \mathbf{a}^H(\theta_k) + \rho_{\text{SI}} \eta \mathbf{g}_k \mathbf{g}_k^H ) \nonumber \\
    = &\nabla_{\mathbf{Q}_{k,i}} \gamma_{k,i} \big|_{\mathbf{Q}_{k,i} = \mathbf{Q}_{k,i}^{[\infty]}}.
  \end{align}
\end{subequations}
Following a similar path, \eqref{eqn:proof_equality_3} can also be obtained. Now, we are ready to prove \textbf{Proposition \ref{proposition_1}}. We take full-duplex as an example, where the Lagrangian function of problem \eqref{problem:narrow_3} as $t \rightarrow \infty$ is given by:
\begin{align}
  \mathcal{L} = &- f(w, r_{k,i}, d_{k,i}) - \sum_{k =1}^2 \sum_{i =1}^2 \lambda_{1, k,i}( r_{k,i}- \gamma_{k,i}^{[\infty]} ) \nonumber \\
   &- \sum_{k =1}^2 \sum_{i =1}^2 \lambda_{2, k,i}( d_{k,i} - D_{k,i}^{[\infty]} ) + \mathcal{C},
\end{align} 
where 
\begin{align}
  \mathcal{C} = &\sum_{k =1}^2 \lambda_{3,k} \left( \sum_{i =1}^2 \mathrm{tr}( \mathbf{Q}_{k,i} ) - 2 \right) \nonumber \\ &- \sum_{k =1}^2 \sum_{i =1}^2 \mathrm{tr} \left( \mathbf{\Upsilon}_{1,k,i} ( \mathbf{Q}_{k,i} - \mathbf{w}_{k,i}\mathbf{w}_{k,i}^H ) \right) \nonumber\\
  &+ \sum_{k =1}^2 \mathrm{tr}\left( \mathbf{\Upsilon}_{2,k}( \mathbf{Q}_{k,1} - \mathbf{Q}_{k,2} ) \right),
\end{align}
and $\boldsymbol{\lambda} = \{ \lambda_{1, k,i}, \lambda_{2, k,i}, \lambda_{3, k} \}$ and $\mathbf{\Upsilon} = \{ \mathbf{\Upsilon}_{1,k,i}, \mathbf{\Upsilon}_{2,k} \}$ denote the Lagrangian dual variables. Since problem \eqref{problem:narrow_3} is convex, its globally optimal solution $\boldsymbol{\chi}^{[\infty]}$ as $t \rightarrow \infty$ must satisfy the KKT conditions of it. Thus, there must exist $\boldsymbol{\lambda}$ and $\mathbf{\Upsilon}$, such that \eqref{KKT_1} on the top of the next page hold. 
By substituting \eqref{eqn:proof_equality_1}-\eqref{eqn:proof_equality_3} into \eqref{KKT_1}, it can be readily verified that \eqref{KKT_1}-\eqref{KKT_2} are exactly the KKT conditions of problem \eqref{problem:narrow_1}. In other words,  $\boldsymbol{\chi}^{[\infty]}$ as $t \rightarrow \infty$ satisfies the KKT conditions of problem \eqref{problem:narrow_1} in full-duplex. The proof corresponding to the half-duplex can also be proved following the same path. The proof is thus completed. 

\section{Derivation of FIM in Wideband Systems} 
We first calculate the covariance matrices $\mathbf{R}_{\mathbf{X}_k} = \mathbb{E}[\mathrm{vec}(\mathbf{X}_k\mathrm{diag}( \mathbf{d}(\nu_k) ) ) \mathrm{vec}(\mathbf{X}_k \mathrm{diag}( \mathbf{d}(\nu_k) ))^H]$ and $\mathbf{R}_{\mathbf{z}} = \mathbb{E}[\mathbf{z} \mathbf{z}^H]$. In contrast to the case of narrowband systems in Appendix A, the cross-correlation of the transmit signal with DAM for the wideband transmission, i.e., $\mathbb{E}[\mathbf{x}_{k,i}^{(q)}[n] (\mathbf{x}_{k,i}^{(q)})^H[n']],$ $\forall n \neq n'$, may not equal to zero due to the delay pre-compensation. Hence, the covariance matrix $\mathbf{R}_{\mathbf{X}_k}$ can be expressed as 
\begin{equation}
  \mathbf{R}_{\mathbf{X}_k} = \mathbf{R}_{\mathbf{X}_k}^{auto} + \mathbf{R}_{\mathbf{X}_k}^{cross},
\end{equation} 
where $\mathbf{R}_{\mathbf{X}_k}^{auto}$ and $\mathbf{R}_{\mathbf{X}_k}^{cross}$ denote the autocorrelation and cross-correlation components of the covariance matrix $\mathbf{R}_{\mathbf{X}_k}$, respectively. In particular, the autocorrelation component $\mathbf{R}_{\mathbf{X}_k}^{auto}$ is
\begin{equation}
  \mathbf{R}_{\mathbf{X}_k}^{auto} =  \mathbf{I}_U \otimes \begin{bmatrix}
    \mathbf{I}_{N_0 /2} \otimes \mathbf{Q}_{k,1} & \mathbf{0} \\
    \mathbf{0} & \mathbf{I}_{N_0 /2} \otimes \mathbf{Q}_{k,2}
  \end{bmatrix},
\end{equation} 
and the cross-correlation component $\mathbf{R}_{\mathbf{X}_k}^{cross}$ satisfies that 
\begin{equation} \label{eqn:blkdiag_1}
  \mathrm{blkdiag}_{M}(\mathbf{R}_{\mathbf{X}_k}^{cross}) = \mathbf{0},
\end{equation}  
where $\mathrm{blkdiag}_{M}(\mathbf{R}_{\mathbf{X}_k}^{cross})$ denotes all the $M \times M$ blocks on the diagonal of matrix $\mathbf{R}_{\mathbf{X}_k}^{cross}$. However, since each element of $\mathbf{z}_k$ is i.i.d. distributed, the cross-correlation component of the covariance of $\mathbf{R}_{\mathbf{z}_k}$ is still zero. Thus, $\mathbf{R}_{\mathbf{z}_k}$ is given by 
\begin{equation} \label{eqn:blkdiag_2}
  \mathbf{R}_{\mathbf{z}_k} = \mathbf{I}_U \otimes \begin{bmatrix}
    (\rho_{\text{SI}} \Phi_{k,1} + 1)\mathbf{I}_{N_0 /2} & \mathbf{0} \\
    \mathbf{0} & (\rho_{\text{SI}} \Phi_{k,2} + 1)\mathbf{I}_{N_0 /2} 
  \end{bmatrix}.
\end{equation}
Then, based on the results in Appendix A, we derive $J_k$ as follows
\begin{align}
  J_k = &2 |\alpha_k|^2 \rho_s  \mathrm{tr}\left( (\mathbf{I}_{2N} \otimes \dot{\mathbf{a}}(\theta_k))  \mathbf{R}_{\mathbf{z}_k}^{-1} (\mathbf{I}_{2N} \otimes \dot{\mathbf{a}}^H(\theta_k)) \mathbf{R}_{\mathbf{X}_k} \right) \nonumber \\
  \overset{(a)}{=} &2 |\alpha_k|^2 \rho_s \mathrm{tr}\left( (\mathbf{I}_{2N} \otimes \dot{\mathbf{a}}(\theta_k))  \mathbf{R}_{\mathbf{z}_k}^{-1} (\mathbf{I}_{2N} \otimes \dot{\mathbf{a}}^H(\theta_k)) \mathbf{R}_{\mathbf{X}_k}^{auto} \right) \nonumber \\
  = &2 |\alpha_k|^2 \rho_s N \sum_{i =1}^2 \frac{\dot{\mathbf{a}}^H(\theta_k) \mathbf{Q}_{k,i} \dot{\mathbf{a}}(\theta_k)}{\rho_{\text{SI}} \Phi_{k,i} + 1}. 
\end{align}
The equality $(a)$ stems from \eqref{eqn:blkdiag_1}, which results in the following equality:
\begin{equation}
  \mathrm{tr} \left( (\mathbf{I}_{2N} \otimes \dot{\mathbf{a}}(\theta_k))  \mathbf{R}_{\mathbf{z}_k}^{-1} (\mathbf{I}_{2N} \otimes \dot{\mathbf{a}}^H(\theta_k)) \mathbf{R}_{\mathbf{X}_k}^{cross} \right) = 0.
\end{equation}
Then, the expression of the expected CRB in \eqref{CRB_wideband} can be obtained.

\end{appendices}

\bibliographystyle{IEEEtran}
\bibliography{reference/mybib}

\begin{thebibliography}{10}
\providecommand{\url}[1]{#1}
\csname url@samestyle\endcsname
\providecommand{\newblock}{\relax}
\providecommand{\bibinfo}[2]{#2}
\providecommand{\BIBentrySTDinterwordspacing}{\spaceskip=0pt\relax}
\providecommand{\BIBentryALTinterwordstretchfactor}{4}
\providecommand{\BIBentryALTinterwordspacing}{\spaceskip=\fontdimen2\font plus
\BIBentryALTinterwordstretchfactor\fontdimen3\font minus
  \fontdimen4\font\relax}
\providecommand{\BIBforeignlanguage}[2]{{%
\expandafter\ifx\csname l@#1\endcsname\relax
\typeout{** WARNING: IEEEtran.bst: No hyphenation pattern has been}%
\typeout{** loaded for the language `#1'. Using the pattern for}%
\typeout{** the default language instead.}%
\else
\language=\csname l@#1\endcsname
\fi
#2}}
\providecommand{\BIBdecl}{\relax}
\BIBdecl

\bibitem{wang2022conference}
Z.~Wang, X.~Mu, and Y.~Liu, ``Full-duplex communication in bidirectional
  {ISAC},'' in \emph{Proc. {IEEE} {ICC} {W}orkshop}, Jun. 2023, pp. 1--6.

\bibitem{zhang20196g}
Z.~Zhang, Y.~Xiao, Z.~Ma, M.~Xiao, Z.~Ding, X.~Lei, G.~K. Karagiannidis, and
  P.~Fan, ``6{G} wireless networks: Vision, requirements, architecture, and key
  technologies,'' \emph{{IEEE} Trans. Veh. Mag.}, vol.~14, no.~3, pp. 28--41,
  Sep. 2019.

\bibitem{samsung20206g}
{Samsung Research}, ``6{G}: The next hyper connected experience for all,''
  Samsung, While Paper, 2020. [Online]. Available:
  {https://research.samsung.com/next-generation-communications}.

\bibitem{sabharwal2014band}
A.~Sabharwal, P.~Schniter, D.~Guo, D.~W. Bliss, S.~Rangarajan, and R.~Wichman,
  ``In-band full-duplex wireless: Challenges and opportunities,'' \emph{{IEEE}
  J. Sel. Areas Commun.}, vol.~32, no.~9, pp. 1637--1652, Sep. 2014.

\bibitem{riihonen2011mitigation}
T.~Riihonen, S.~Werner, and R.~Wichman, ``Mitigation of loopback
  self-interference in full-duplex {MIMO} relays,'' \emph{{IEEE} Trans. Signal
  Process.}, vol.~59, no.~12, pp. 5983--5993, Dec. 2011.

\bibitem{bharadia2013full}
D.~Bharadia, E.~McMilin, and S.~Katti, ``Full duplex radios,'' in \emph{Proc.
  ACM SIGCOMM}, Aug. 2013, pp. 375--386.

\bibitem{everett2014passive}
E.~Everett, A.~Sahai, and A.~Sabharwal, ``Passive self-interference suppression
  for full-duplex infrastructure nodes,'' \emph{{IEEE} Trans. Wireless
  Commun.}, vol.~13, no.~2, pp. 680--694, Feb. 2014.

\bibitem{zhang2015full}
Z.~Zhang, X.~Chai, K.~Long, A.~V. Vasilakos, and L.~Hanzo, ``Full duplex
  techniques for {5G} networks: self-interference cancellation, protocol
  design, and relay selection,'' \emph{{IEEE} Commun. Mag.}, vol.~53, no.~5,
  pp. 128--137, May 2015.

\bibitem{sim2017nonlinear}
M.~S. Sim, M.~Chung, D.~Kim, J.~Chung, D.~K. Kim, and C.-B. Chae, ``Nonlinear
  self-interference cancellation for full-duplex radios: From link-level and
  system-level performance perspectives,'' \emph{{IEEE} Commun. Mag.}, vol.~55,
  no.~9, pp. 158--167, Sep. 2017.

\bibitem{khaledian2018inherent}
S.~Khaledian, F.~Farzami, B.~Smida, and D.~Erricolo, ``Inherent
  self-interference cancellation for in-band full-duplex single-antenna
  systems,'' \emph{{IEEE} Trans. Microw. Theory Techn.}, vol.~66, no.~6, pp.
  2842--2850, Jun. 2018.

\bibitem{zhang2020perceptive}
A.~Zhang, M.~L. Rahman, X.~Huang, Y.~J. Guo, S.~Chen, and R.~W. Heath,
  ``Perceptive mobile networks: Cellular networks with radio vision via joint
  communication and radar sensing,'' \emph{IEEE Veh. Technol. Mag.}, vol.~16,
  no.~2, pp. 20--30, Jun. 2021.

\bibitem{liu2022integrated}
F.~Liu, Y.~Cui, C.~Masouros, J.~Xu, T.~X. Han, Y.~C. Eldar, and S.~Buzzi,
  ``Integrated sensing and communications: Towards dual-functional wireless
  networks for {6G} and beyond,'' \emph{{IEEE} J. Sel. Areas Commun.}, vol.~40,
  no.~6, pp. 1728--1767, Mar. 2022.

\bibitem{tong2022environment}
X.~Tong, Z.~Zhang, Y.~Zhang, Z.~Yang, C.~Huang, K.-K. Wong, and M.~Debbah,
  ``Environment sensing considering the occlusion effect: A multi-view
  approach,'' \emph{{IEEE} Trans. Signal Process.}, vol.~70, pp. 3598--3615,
  Jun. 2022.

\bibitem{day2012full}
B.~P. Day, A.~R. Margetts, D.~W. Bliss, and P.~Schniter, ``Full-duplex
  bidirectional {MIMO}: Achievable rates under limited dynamic range,''
  \emph{{IEEE} Trans. Signal Process.}, vol.~60, no.~7, pp. 3702--3713, Jul.
  2012.

\bibitem{taghizadeh2018hardware}
O.~Taghizadeh, V.~Radhakrishnan, A.~C. Cirik, R.~Mathar, and L.~Lampe,
  ``Hardware impairments aware transceiver design for bidirectional full-duplex
  {MIMO OFDM} systems,'' \emph{{IEEE} Trans. Veh. Technol.}, vol.~67, no.~8,
  pp. 7450--7464, Aug. 2018.

\bibitem{da2016spectral}
J.~M.~B. da~Silva, G.~Fodor, and C.~Fischione, ``Spectral efficient and fair
  user pairing for full-duplex communication in cellular networks,''
  \emph{{IEEE} Trans. Wireless Commun.}, vol.~15, no.~11, pp. 7578--7593, Nov.
  2016.

\bibitem{sun2016multi}
Y.~Sun, D.~W.~K. Ng, J.~Zhu, and R.~Schober, ``Multi-objective optimization for
  robust power efficient and secure full-duplex wireless communication
  systems,'' \emph{{IEEE} Trans. Wireless Commun.}, vol.~15, no.~8, pp.
  5511--5526, Aug. 2016.

\bibitem{day2012full_relay}
B.~P. Day, A.~R. Margetts, D.~W. Bliss, and P.~Schniter, ``Full-duplex {MIMO}
  relaying: Achievable rates under limited dynamic range,'' \emph{{IEEE} J.
  Sel. Areas Commun.}, vol.~30, no.~8, pp. 1541--1553, Sep. 2012.

\bibitem{ng2012dynamic}
D.~W.~K. Ng, E.~S. Lo, and R.~Schober, ``Dynamic resource allocation in
  {MIMO-OFDMA} systems with full-duplex and hybrid relaying,'' \emph{{IEEE}
  Trans. Wireless Commun.}, vol.~60, no.~5, pp. 1291--1304, May 2012.

\bibitem{suraweera2014low}
H.~A. Suraweera, I.~Krikidis, G.~Zheng, C.~Yuen, and P.~J. Smith,
  ``Low-complexity end-to-end performance optimization in {MIMO} full-duplex
  relay systems,'' \emph{{IEEE} Trans. Wireless Commun.}, vol.~13, no.~2, pp.
  913--927, Feb. 2014.

\bibitem{ngo2014multipair}
H.~Q. Ngo, H.~A. Suraweera, M.~Matthaiou, and E.~G. Larsson, ``Multipair
  full-duplex relaying with massive arrays and linear processing,''
  \emph{{IEEE} J. Sel. Areas Commun.}, vol.~32, no.~9, pp. 1721--1737, Sep.
  2014.

\bibitem{liu2018toward}
F.~Liu, L.~Zhou, C.~Masouros, A.~Li, W.~Luo, and A.~Petropulu, ``Toward
  dual-functional radar-communication systems: Optimal waveform design,''
  \emph{{IEEE} Trans. Signal Process.}, vol.~66, no.~16, pp. 4264--4279, Aug.
  2018.

\bibitem{liu2020joint}
X.~Liu, T.~Huang, N.~Shlezinger, Y.~Liu, J.~Zhou, and Y.~C. Eldar, ``Joint
  transmit beamforming for multiuser {MIMO} communications and {MIMO} radar,''
  \emph{{IEEE} Trans. Signal Process.}, vol.~68, pp. 3929--3944, Jun. 2020.

\bibitem{pritzker2022transmit}
J.~Pritzker, J.~Ward, and Y.~C. Eldar, ``Transmit precoder design approaches
  for dual-function radar-communication systems,'' \emph{arXiv preprint
  arXiv:2203.09571}, 2022.

\bibitem{zhang2022holographic}
H.~Zhang, H.~Zhang, B.~Di, M.~Di~Renzo, Z.~Han, H.~V. Poor, and L.~Song,
  ``Holographic integrated sensing and communication,'' \emph{{IEEE} J. Sel.
  Areas Commun.}, vol.~40, no.~7, pp. 2114--2130, Jul. 2022.

\bibitem{liu2022joint}
R.~Liu, M.~Li, Q.~Liu, and A.~L. Swindlehurst, ``Joint waveform and filter
  designs for {STAP-SLP}-based {MIMO-DFRC} systems,'' \emph{{IEEE} J. Sel.
  Areas Commun.}, vol.~40, no.~6, pp. 1918--1931, 2022.

\bibitem{liu2021cramer}
F.~Liu, Y.-F. Liu, A.~Li, C.~Masouros, and Y.~C. Eldar, ``Cram{\'e}r-rao bound
  optimization for joint radar-communication beamforming,'' \emph{{IEEE} Trans.
  Signal Process.}, vol.~70, pp. 240--253, Dec. 2021.

\bibitem{wang2022stars}
Z.~Wang, X.~Mu, and Y.~Liu, ``{STARS} enabled integrated sensing and
  communications,'' \emph{{IEEE} Trans. Wireless Commun.}, vol.~22, no.~10, pp.
  6750--6765, Oct. 2023.

\bibitem{li2022novel}
S.~Li, W.~Yuan, C.~Liu, Z.~Wei, J.~Yuan, B.~Bai, and D.~W.~K. Ng, ``A novel
  {ISAC} transmission framework based on spatially-spread orthogonal time
  frequency space modulation,'' \emph{{IEEE} J. Sel. Areas Commun.}, vol.~40,
  no.~6, pp. 1854--1872, Jun. 2022.

\bibitem{xiao2022integrated}
Z.~Xiao and Y.~Zeng, ``Integrated sensing and communication with delay
  alignment modulation: Performance analysis and beamforming optimization,''
  \emph{{IEEE} Trans. Wireless Commun.}, Early Access, Apr. 2023. doi:
  10.1109/TWC.2023.3266758.

\bibitem{tse2005fundamentals}
D.~Tse and P.~Viswanath, \emph{Fundamentals of wireless communication}.\hskip
  1em plus 0.5em minus 0.4em\relax Cambridge, U.K.: Cambridge Univ. Press,
  2005.

\bibitem{skolnik1980introduction}
M.~I. Skolnik, \emph{Introduction to Radar Systems}.\hskip 1em plus 0.5em minus
  0.4em\relax New York, NY, USA: McGraw-Hill, 1990.

\bibitem{kay1993fundamentals}
S.~M. Kay, \emph{Fundamentals of statistical signal processing: estimation
  theory}.\hskip 1em plus 0.5em minus 0.4em\relax Englewood Cliffs, NJ:
  Prentice-Hall, 1993.

\bibitem{marler2004survey}
R.~T. Marler and J.~S. Arora, ``Survey of multi-objective optimization methods
  for engineering,'' \emph{Struct. Multidisciplinary Optim.}, vol.~26, no.~6,
  pp. 369--395, 2004.

\bibitem{huang2010dual}
Y.~Huang and D.~P. Palomar, ``A dual perspective on separable semidefinite
  programming with applications to optimal downlink beamforming,'' \emph{{IEEE}
  Trans. Signal Process.}, vol.~58, no.~8, pp. 4254--4271, Aug. 2010.

\bibitem{lu2022delay}
H.~Lu and Y.~Zeng, ``Delay alignment modulation: Enabling equalization-free
  single-carrier communication,'' \emph{{IEEE} Wireless Commun. Lett.},
  vol.~11, no.~9, pp. 1785--1789, Sep. 2022.

\bibitem{cvx}
M.~Grant and S.~Boyd, ``{CVX}: Matlab software for disciplined convex
  programming, version 2.1,'' \url{http://cvxr.com/cvx}, Mar. 2014.

\bibitem{mckay2006capacity}
M.~R. McKay and I.~B. Collings, ``On the capacity of frequency-flat and
  frequency-selective {Rician} {MIMO} channels with single-ended correlation,''
  \emph{{IEEE} Trans. Wireless Commun.}, vol.~5, no.~8, pp. 2038--2043, Aug.
  2006.

\end{thebibliography}

\end{document}